\documentclass[journal,twoside,web]{ieeecolor}
\usepackage{generic}

\usepackage{cite}

\usepackage{amssymb,amsmath,latexsym,amsfonts,amsthm,mathtools}

\usepackage[colorlinks=true]{hyperref}
\hypersetup{colorlinks, breaklinks, citecolor=blue, linkcolor=blue, urlcolor=blue}
\usepackage{graphicx}
\usepackage{float,subcaption}
\usepackage{textcomp}
\usepackage{xcolor}
\definecolor{darkpastelpurple}{rgb}{0.59, 0.44, 0.84}
\usepackage[nameinlink]{cleveref}
\Crefformat{figure}{#2Fig.~#1#3}
\Crefmultiformat{figure}{Figs.~#2#1#3}{ and~#2#1#3}{, #2#1#3}{ and~#2#1#3}

\usepackage{tikz}
\usetikzlibrary{automata, shapes, arrows, calc, arrows.meta, fit, positioning}

\theoremstyle{plain}
\newtheorem{theorem}{Theorem}
\newtheorem{lemma}{Lemma}
\newtheorem{corollary}{Corollary}
\newtheorem{proposition}{Proposition}
\newtheorem*{problem*}{Problem}
\newtheorem{problem}{Problem}
\Crefname{problem}{Problem}{Problems} 
\crefname{problem}{problem}{problems}     

\theoremstyle{remark}
\newtheorem{remark}{Remark}
\newtheorem{assumption}{Assumption}
\Crefname{assumption}{Assumption}{Assumptions} 
\crefname{assumption}{assumption}{assumptions}     

\newtheorem{definition}{Definition}
\theoremstyle{definition}

\usepackage{mathrsfs}

\usepackage{newtxmath}
\usepackage{booktabs}
\usepackage{duckuments}

\begin{document}

\title{Distributed Safe Consensus Under Asymmetric Input and Time-Varying Output Constraints}
	\author{Abhinav Sinha,~\IEEEmembership{Senior Member,~IEEE}, and Shashi Ranjan Kumar,~\IEEEmembership{Senior Member,~IEEE}
		\thanks{A. Sinha is with the Guidance, Autonomy, Learning, and Control for Intelligent Systems (GALACxIS) Lab, Department of Aerospace Engineering and Engineering Mechanics, University of Cincinnati, OH 45221, USA. (e-mail: abhinav.sinha@uc.edu). S. R. Kumar is with the Intelligent Systems and Control (ISaC) Lab, Department of Aerospace Engineering, Indian Institute of Technology Bombay, Powai, Mumbai 400076, India. (email: srk@aero.iitb.ac.in).}
	}

	\maketitle

    \begin{abstract}
This paper studies safe distributed consensus for single-integrator multi-agent systems over connected undirected graphs under simultaneous asymmetric actuator constraints and output safety constraints. Each agent is equipped with a continuously differentiable asymmetric actuator dynamics that maps a commanded control signal to the realized plant input while keeping the latter strictly inside a prescribed admissible interval. To address output safety, a barrier-coordinate transformation is introduced over a common time-varying safe interval, and a distributed synchronization law is designed in the transformed coordinates. The resulting controller integrates a graph-based coordination layer with an actuator-side tracking layer, thereby enabling simultaneous enforcement of input admissibility, forward invariance of the safe output set, and asymptotic synchronization. For compact admissible sets of initial conditions, it is shown that the closed-loop solution is complete, all signals remain bounded, the actuator inputs remain strictly within their asymmetric bounds, and the agent outputs remain inside the prescribed safe interval for all time. Moreover, the transformed synchronization errors converge exponentially to zero, and the original agent outputs asymptotically synchronize to a designer-selected admissible trajectory embedded in the common safe interval. Numerical simulations validate the proposed framework and demonstrate safe consensus under both asymmetric actuation bounds and time-varying output constraints.
\end{abstract}


    \section{Introduction}\label{sec:introduction}
    \IEEEPARstart{D}{istributed} consensus is crucial in networked control systems and multi-agent coordination. Over the past two decades, consensus theory \cite{OlfatiSaber2007,OlfatiSaberMurray2004,RenBeard2005,mesbahi2010graph,9072294} has developed into a mature field. Despite these advances, much of the existing literature remains idealized from the standpoint of implementation, i.e., the control channel is often assumed to be unconstrained or only implicitly constrained, while safety requirements on the agent outputs are either neglected or treated separately from the consensus objective \cite{OlfatiSaber2007,Ames2019,Wang2017}.

    In practical multi-agent systems, the control action is delivered through physical actuators with finite capabilities, and these capabilities are often asymmetric. For instance, acceleration and braking authority may differ in ground vehicles, thrust and reverse-thrust limits may differ in marine platforms, and positive and negative rate limits may be inherently unequal in aerospace and robotic systems \cite{actuator_constraints_1,actuator_constraints_2,actuator_constraints_3}. If such asymmetries are ignored during controller synthesis, the resulting distributed law may request infeasible inputs, degrade performance, or even invalidate the closed-loop guarantees. At the same time, many safety-critical tasks require each agent to evolve within a prescribed admissible output interval or corridor, either because of operational envelopes, mission-level safety requirements, or environment-induced restrictions \cite{Tee2009,safe_coordination_2,ranjan2026engagement,Wang2017,9149651,zhang2025tcns_state_constraints}. Thus, a practically meaningful safe-consensus design must enforce both input admissibility and output safety while preserving distributed coordination.

    Existing results address important parts of this problem, but not the complete unified treatment considered here. Consensus algorithms with saturated or bounded inputs have established valuable coordination guarantees, yet they often rely on symmetric saturation descriptions, nonsmooth clipping operators, or input constraints that are not embedded directly into the plant-actuator dynamics \cite{8950207,10795196,Yang2014,Yi2019,zhang2023neural,SAJJADIKIA2013762,10982103}. On the other hand, barrier-based and safety-critical multi-agent methods provide powerful tools for set invariance and collision avoidance, but actuator admissibility is often treated implicitly, through external filtering, or as an auxiliary layer appended to a nominal controller \cite{sharifi2023tcns_hocbf_leader_follower,song2025tcns_safety_flocking,zhang2025tcns_state_constraints,10794772,Ames2019,Wang2017,10556645,zhang2023neural,11499442}. Consequently, a unified distributed design that simultaneously guarantees strict asymmetric actuator admissibility, time-varying output safety, and asymptotic synchronization remains an open and practically relevant problem.

    Motivated by this gap, this paper develops a distributed safe-consensus framework for single-integrator multi-agent systems over connected undirected graphs under simultaneous asymmetric input constraints and time-varying output constraints. The proposed architecture blends an actuator-side tracking layer with a graph-based coordination law. For output safety, a logarithmic barrier-coordinate transformation is introduced over a common admissible core interval. In the transformed coordinates, a distributed synchronization law with a pinning term drives the agents toward a prescribed admissible reference trajectory embedded inside the safe interval. The proposed design treats actuator feasibility and output safety as integral components of the distributed controller architecture and simultaneously achieves: \emph{(i)} strict enforcement of agent-wise asymmetric actuator bounds through a continuously differentiable asymmetric actuator dynamics, \emph{(ii)} rigorous forward invariance of a common time-varying safe interval, and \emph{(iii)} asymptotic synchronization to a designer-selected admissible trajectory using only local graph interactions. 

    The proposed approach is also distinct from funnel-control and prescribed-performance methods \cite{10794772,10982103}. Funnel-control designs typically constrain tracking or formation errors to remain within prescribed transient envelopes, thereby shaping the evolution of an error variable. However, the present work imposes safety directly on the agent outputs through a common time-varying admissible core and then performs distributed synchronization in the corresponding barrier coordinates. Moreover, the proposed design assigns a designer-selected admissible trajectory and proves asymptotic synchronization to this trajectory, rather than only maintaining an error inside a prescribed funnel. Most importantly, the present framework treats asymmetric actuator limits as part of the plant-actuator dynamics and proves the strict admissibility of the realized inputs together with the forward invariance of the output safe set. Thus, input feasibility, output safety, and distributed synchronization are addressed in a single closed-loop analysis.  The main contributions of this paper are as follows. 
    \begin{itemize}
        \item We formulate a distributed consensus problem for single-integrator multi-agent systems under simultaneous asymmetric input constraints and time-varying output safety constraints, where the actuator bounds are embedded directly into the agent dynamics through a continuously differentiable asymmetric actuator dynamics.
        \item We develop a distributed controller for the input-constrained case that combines a graph-based nominal consensus law with an actuator-tracking layer and guarantees completeness of solutions, boundedness of all closed-loop signals, strict satisfaction of asymmetric actuator bounds, and asymptotic consensus on compact admissible initial sets.
        \item For the safe-consensus problem, we introduce a barrier-coordinate distributed synchronization law over a common time-varying admissible interval and establish forward invariance of the safe set, strict input admissibility, exponential convergence of the transformed synchronization error, and asymptotic synchronization of the original agent outputs. We show that incorporating a pinning term in the transformed coordinates allows the proposed design to achieve synchronization to a designer-selected admissible trajectory inside the common safe interval, which is stronger than consensus to an unspecified limit.
        \item The analysis is developed in a self-contained semiglobal framework that avoids auxiliary \emph{a priori} boundedness assumptions and yields explicit gain and interiority conditions ensuring well-posedness and safety of the closed-loop system.
    \end{itemize}

    \section{Problem Formulation}\label{sec:problem_formulation}
Let $\mathbb{R}$, $\mathbb{R}_{>0}$, and $\mathbb{R}_{\ge 0}$ denote the sets of real, positive real, and nonnegative real numbers, respectively. For a vector $\mathbf{x} \in \mathbb{R}^{N}$, $\|\mathbf{x}\|$ denotes the Euclidean norm. The symbols $\mathbf{1}_{N} \in \mathbb{R}^{N}$ and $\mathbf{0}_{N} \in \mathbb{R}^{N}$ denote the vectors of all ones and all zeros, respectively, while $I_{N}$ denotes the $N \times N$ identity matrix. For scalars $a_{i}$, $i \in \{1,\dots,N\}$, $\mathrm{col}\{a_{1},\dots,a_{N}\}$ denotes the stacked column vector and $\mathrm{diag}\{a_{1},\dots,a_{N}\}$ denotes the corresponding diagonal matrix.

Consider a weighted undirected graph $\mathcal{G} = \left(\mathcal{V},\mathcal{E},A\right)$, where $\mathcal{V} = \{1,2,\dots,N\}$ is the node set, $\mathcal{E} \subseteq \mathcal{V} \times \mathcal{V}$ is the edge set, and $A = [a_{ij}] \in \mathbb{R}^{N \times N}$ is the adjacency matrix satisfying $a_{ij} = a_{ji} \ge 0,~a_{ii} = 0$, and $a_{ij} > 0 \iff (i,j) \in \mathcal{E}$. The neighbor set of agent $i$ is defined by $\mathcal{N}_{i} = \{j \in \mathcal{V} : (i,j) \in \mathcal{E}\}$. The degree matrix and Laplacian matrix associated with $\mathcal{G}$ are defined as $D = \mathrm{diag}\{d_{1},\dots,d_{N}\}$ with $d_{i} = \sum_{j=1}^{N} a_{ij}$ and $L=D-A$. Throughout the paper, the following standing assumption is imposed.
\begin{assumption}\label{ass:connected_graph}
The graph $\mathcal{G}$ is undirected and connected.
\end{assumption}
Under \Cref{ass:connected_graph}, the Laplacian satisfies $L = L^{\top} \ge 0$, $L \mathbf{1}_{N} = \mathbf{0}_{N}$, and its eigenvalues can be ordered as $0 = \lambda_{1}(L) < \lambda_{2}(L) \le \cdots \le \lambda_{N}(L)$. Define the consensus subspace $\mathcal{C} := \mathrm{span}\{\mathbf{1}_{N}\}
    = \{\mathbf{x} \in \mathbb{R}^{N} : x_{1} = \cdots = x_{N}\}$, and the orthogonal projection onto $\mathcal{C}^{\perp}$ by $\Pi := I_{N} - \frac{1}{N}\mathbf{1}_{N}\mathbf{1}_{N}^{\top}$. For any stacked state vector $\mathbf{x} = \mathrm{col}\{x_{1},\dots,x_{N}\}$, the disagreement vector is defined as $\tilde{\mathbf{x}} := \Pi \mathbf{x}$. Clearly, $\tilde{\mathbf{x}} = \mathbf{0}_{N}$ if and only if $\mathbf{x} \in \mathcal{C}$.

Consider a network of $N$ single-integrator agents whose plant dynamics are given by
\begin{align}
    \dot{x}_{i} = u_{i}, ~~ y_{i} = x_{i}, ~~ i \in \mathcal{V}, \label{eq:agent_dynamics}
\end{align}
where $x_{i} \in \mathbb{R}$ is the state of agent $i$, $y_{i} \in \mathbb{R}$ is the measured output\footnote{Since $y_i=x_i$ in \eqref{eq:agent_dynamics}, output constraints and state constraints are identical in the present setting. Therefore, these terms are used interchangeably whenever no confusion arises.}, and $u_{i} \in \mathbb{R}$ is the actual input delivered to the plant.

To explicitly account for actuator limitations \cite{Kumar2025}, we distinguish between the commanded signal $v_{i}$ generated by the controller and the actual plant input $u_{i}$. For each agent $i$, the actuator is modeled by an asymmetric actuator dynamics whose regularization factor is continuously
differentiable on the admissible input interval and vanishes only at the
actuator boundaries
\begin{align}
    \dot{u}_{i}
    =
    p_{1,i}\,\sigma_{i}(u_{i})\,v_{i}
    -
 p_{2,i} u_{i}, \label{eq:actuator_dynamics}
\end{align}
where $p_{1,i} \in \mathbb{R}_{>0}$, $p_{2,i} \in \mathbb{R}_{>0}$, and
\begin{align}
    \sigma_{i}(u_{i})
    :=
    \varrho_{i}
    \left(
        1 - \left(\frac{u_{i}}{\overline{u}_{i}}\right)^{\gamma_{i}}
    \right)
    +
    (1-\varrho_{i})
    \left(
        1 - \left(\frac{u_{i}}{\underline{u}_{i}}\right)^{\gamma_{i}}
    \right), \label{eq:sigma_def}
\end{align}
with $\varrho_{i}
    :=1$ if $u_{i}>0$ and $0$ otherwise, and $\gamma_{i} = 2n_i~\forall~n_i\in\mathbb{N}$. Here, $\underline{u}_{i} < 0 < \overline{u}_{i}$ denote the lower and upper admissible actuator limits of agent $i$.
\begin{remark}\label{rem:asymmetric_input_limits}
In general, the actuator bounds are asymmetric, that is, $|\underline{u}_i|\neq |\overline{u}_i|$. Therefore, the available control authority in the positive and negative directions need not be identical.
\end{remark}
For each agent $i$, define the admissible input set $\mathcal{U}_{i}
    :=
    \{u \in \mathbb{R} : \underline{u}_{i} < u < \overline{u}_{i}\}$, and the network-level admissible input set $\mathcal{U}
    :=
    \mathcal{U}_{1} \times \cdots \times \mathcal{U}_{N}$. The stacked plant and actuator states are denoted by $\mathbf{x} := \mathrm{col}\{x_{1},\dots,x_{N}\}$, $\mathbf{u} := \mathrm{col}\{u_{1},\dots,u_{N}\}$, and $\mathbf{v} := \mathrm{col}\{v_{1},\dots,v_{N}\}$. The following assumption ensures that the actuator dynamics are well posed and that the initial actuator states are admissible.
\begin{assumption}\label{ass:actuator}
For each $i \in \mathcal{V}$, the constants $p_{1,i}$, $p_{2,i}$, $\underline{u}_{i}$, $\overline{u}_{i}$, and $\gamma_{i}$ satisfy 
    $p_{1,i} > 0,~ p_{2,i} > 0, ~ \underline{u}_{i} < 0 < \overline{u}_{i},~ \gamma_{i} = 2n_i$,
and the initial actuator state satisfies $u_{i}(0) \in \mathcal{U}_{i}$.
\end{assumption}
In addition to actuator constraints, each agent is required to satisfy a possibly time-varying output constraint of the form $x_{i}(t) \in \mathcal{X}_{i}(t), ~ \forall t \ge 0$, where    $\mathcal{X}_{i}(t)
    :=
    \{x \in \mathbb{R} : \underline{x}_{i}(t) < x < \overline{x}_{i}(t)\}$,
and $\underline{x}_{i}(\cdot)$ and $\overline{x}_{i}(\cdot)$ denote the lower and upper safety bounds associated with agent $i$. Define the admissible network state set at time $t$ as $\mathcal{X}(t)
    :=
    \mathcal{X}_{1}(t) \times \cdots \times \mathcal{X}_{N}(t)$. Since consensus requires all states to asymptotically agree, a necessary feasibility condition is that the individual admissible intervals admit a nonempty common intersection. To this end, define the consensus-feasible set $\mathcal{X}_{\cap}(t)
    :=
    \bigcap_{i=1}^{N} \mathcal{X}_{i}(t)
    =
    \left(
        \max_{1 \le i \le N} \underline{x}_{i}(t),
        \min_{1 \le i \le N} \overline{x}_{i}(t)
    \right)$.
The following assumption is adopted for the output-constrained consensus problem.
\begin{assumption}\label{ass:output_constraints}
For each $i \in \mathcal{V}$, the functions $\underline{x}_{i} : \mathbb{R}_{\ge 0} \to \mathbb{R}$ and $\overline{x}_{i} : \mathbb{R}_{\ge 0} \to \mathbb{R}$ are continuously differentiable and satisfy $\underline{x}_{i}(t) < \overline{x}_{i}(t), ~ \forall t \ge 0$. Moreover, there exists a constant $\delta_{x} > 0$ such that
    $\min_{1 \le i \le N} \overline{x}_{i}(t)
    -
    \max_{1 \le i \le N} \underline{x}_{i}(t)
    \ge \delta_{x}, ~~ \forall t \ge 0$,
that is, $\mathcal{X}_{\cap}(t) \neq \emptyset$ for all $t \ge 0$.
\end{assumption}
\Cref{ass:output_constraints} is necessary for the safe consensus problem to be meaningful, since asymptotic agreement is impossible if the admissible intervals do not share a common feasible region. We consider dynamic distributed controllers of the form
\begin{align}
    \dot{\eta}_{i}
    &=
    \Phi_{i}
    \left(
        \eta_{i},
        x_{i},
        u_{i},
        \{x_{j} - x_{i}\}_{j \in \mathcal{N}_{i}},
        \{u_{j} - u_{i}\}_{j \in \mathcal{N}_{i}},
        t
    \right), \label{eq:controller_state} \\
    v_{i}
    &=
    \Psi_{i}
    \left(
        \eta_{i},
        x_{i},
        u_{i},
        \{x_{j} - x_{i}\}_{j \in \mathcal{N}_{i}},
        \{u_{j} - u_{i}\}_{j \in \mathcal{N}_{i}},
        t
    \right), \label{eq:controller_output}
\end{align}
where $\eta_{i} \in \mathbb{R}^{q_{i}}$ is the controller state of agent $i$, and the mappings $\Phi_{i}$ and $\Psi_{i}$ are locally Lipschitz in their state arguments and piecewise continuous in time. The stacked controller state is denoted by $\boldsymbol{\eta}
    :=
    \mathrm{col}\{\eta_{1},\dots,\eta_{N}\}$, with total dimension $q := \sum_{i=1}^{N} q_{i}$. The controller class \eqref{eq:controller_state}--\eqref{eq:controller_output} formalizes the distributed-information constraint that each agent may only use its own variables and relative information available from neighboring agents.
\begin{definition}\label{def:complete_solution}
A solution of the closed-loop system consisting of \eqref{eq:agent_dynamics}, \eqref{eq:actuator_dynamics}, and \eqref{eq:controller_state}--\eqref{eq:controller_output} is said to be complete if it is defined on the entire interval $[0,\infty)$.
\end{definition}
\begin{definition}\label{def:consensus}
The closed-loop multi-agent system is said to achieve asymptotic consensus if every complete solution satisfies $\lim_{t \to \infty} \|\Pi \mathbf{x}(t)\| = 0$.
Equivalently, $\lim_{t \to \infty} |x_{i}(t) - x_{j}(t)| = 0,
    ~ \forall i,j \in \mathcal{V}$.
\end{definition}
\begin{definition}\label{def:safe_consensus}
The closed-loop multi-agent system is said to achieve safe consensus if it achieves asymptotic consensus and, in addition, satisfies $\mathbf{u}(t) \in \mathcal{U}, ~ \mathbf{x}(t) \in \mathcal{X}(t), ~ \forall t \ge 0$.
\end{definition}
We now formulate the two problems addressed in this paper.
\begin{problem}[Distributed consensus under asymmetric bounded inputs]\label{prob:input_constrained_consensus}
Consider the networked system \eqref{eq:agent_dynamics}--\eqref{eq:actuator_dynamics} over the graph $\mathcal{G}$ satisfying \Cref{ass:connected_graph,ass:actuator}. Design a distributed controller of the form \eqref{eq:controller_state}--\eqref{eq:controller_output} such that, for every compact set $\mathcal{K} \subset \mathbb{R}^{N} \times \mathcal{U} \times \mathbb{R}^{q}$, there exist controller parameters for which, for all initial conditions $\left(\mathbf{x}(0),\mathbf{u}(0),\boldsymbol{\eta}(0)\right) \in \mathcal{K}$, the following properties hold: (i) the corresponding closed-loop solution is complete; (ii) the input constraints are satisfied for all time, namely $\mathbf{u}(t) \in \mathcal{U}, ~\forall t \ge 0$; (iii) the network achieves asymptotic consensus, that is, $\lim_{t \to \infty} \|\Pi \mathbf{x}(t)\| = 0$; and (iv) all closed-loop signals remain bounded on $[0,\infty)$. When the above properties hold for all admissible initial conditions, the closed-loop system is said to achieve \emph{semiglobal consensus under asymmetric bounded inputs}.
\end{problem}
\begin{problem}[Safe distributed consensus under simultaneous input and output constraints]\label{prob:safe_consensus}
Consider the networked system \eqref{eq:agent_dynamics}--\eqref{eq:actuator_dynamics} over the graph $\mathcal{G}$ satisfying \Crefrange{ass:connected_graph}{ass:output_constraints}. Design a distributed controller of the form \eqref{eq:controller_state}--\eqref{eq:controller_output} such that, for every compact set
    $\mathcal{K}
    \subset
    \mathcal{X}(0) \times \mathcal{U} \times \mathbb{R}^{q}$,
there exist controller parameters for which, for all initial conditions
$\left(\mathbf{x}(0),\mathbf{u}(0),\boldsymbol{\eta}(0)\right) \in \mathcal{K}$,
the following properties hold: (i) the corresponding closed-loop solution is complete; (ii) the input constraints are satisfied for all time, namely $\mathbf{u}(t) \in \mathcal{U}, ~ \forall t \ge 0$; (iii) the output constraints are satisfied for all time, namely $\mathbf{x}(t) \in \mathcal{X}(t), ~ \forall t \ge 0$; (iv) the network achieves asymptotic consensus, that is, $\lim_{t \to \infty} \|\Pi \mathbf{x}(t)\| = 0$; and (v) all closed-loop signals remain bounded on $[0,\infty)$. When the above properties hold for every compact admissible initial set $\mathcal{K}$, the closed-loop system is said to achieve \emph{semiglobal safe consensus under simultaneous input and output constraints}.
\end{problem}
\begin{remark}\label{rem:consensus_not_average}
\Cref{prob:input_constrained_consensus,prob:safe_consensus} concern asymptotic agreement, not necessarily average consensus. Since actuator dynamics \eqref{eq:actuator_dynamics} introduces additional internal dynamics between the commanded signal $v_{i}$ and the realized input $u_{i}$, preservation of the arithmetic average of the initial states is not imposed a priori.
\end{remark}
\begin{remark}\label{rem:semiglobal}
The safe consensus objective in \Cref{prob:safe_consensus} is posed in a semiglobal sense because the state constraints define an open admissible set, and forward invariance must be guaranteed simultaneously with convergence of the network disagreement dynamics.
\end{remark}

\section{Main Results}\label{sec:main_results}
For each agent $i \in \mathcal{V}$, define the nominal consensus input
\begin{align}
    \alpha_{i}(\mathbf{x})
    :=
    -k \sum_{j=1}^{N} a_{ij}(x_{i}-x_{j}),
    ~~
    k \in \mathbb{R}_{>0}, \label{eq:alpha_i_main}
\end{align}
and the actuator-tracking error
\begin{align}
    \varepsilon_{i}
    :=
    u_{i} - \alpha_{i}(\mathbf{x}). \label{eq:epsilon_i_main}
\end{align}
The distributed commanded input is chosen as
\begin{align}
    v_{i}
    =
    \frac{
        p_{2,i}u_{i}
        +
        \dot{\alpha}_{i}(\mathbf{x},\mathbf{u})
        -
        c_{i}\varepsilon_{i}
    }{
        p_{1,i}\sigma_{i}(u_{i})
    },
    ~~
    c_{i} \in \mathbb{R}_{>0}, \label{eq:v_i_main}
\end{align}
where the term $\dot{\alpha}_{i}(\mathbf{x},\mathbf{u})$ is defined as
\begin{align}
    \dot{\alpha}_{i}(\mathbf{x},\mathbf{u})
    =
    -k \sum_{j=1}^{N} a_{ij}(u_{i}-u_{j}). \label{eq:alpha_dot_i_main}
\end{align}
\begin{proposition}\label{prop:reduced_input_dynamics}
Consider the closed-loop system consisting of \eqref{eq:agent_dynamics}, \eqref{eq:actuator_dynamics}, and \eqref{eq:v_i_main}. Let $ \boldsymbol{\varepsilon}
    :=
    \mathrm{col}\{\varepsilon_{1},\dots,\varepsilon_{N}\},
    ~
    \tilde{\mathbf{x}}
    :=
    \Pi \mathbf{x}$. Then, on every interval on which the closed-loop solution is well defined and satisfies $\mathbf{u}(t) \in \mathcal{U}$, one has the following dynamics
\begin{align}
    \dot{\varepsilon}_{i}
    =
    -c_{i}\varepsilon_{i},
    ~~
    i \in \mathcal{V}, \label{eq:eps_dynamics_main}
\end{align}
and
\begin{align}
    \dot{\tilde{\mathbf{x}}}
    =
    -kL\tilde{\mathbf{x}} + \Pi\boldsymbol{\varepsilon}. \label{eq:tilde_x_dynamics_main}
\end{align}
\end{proposition}

\begin{proof}
Differentiating \eqref{eq:epsilon_i_main} with respect to time and using \eqref{eq:actuator_dynamics} together with \eqref{eq:v_i_main} gives
\begin{align}
    \dot{\varepsilon}_{i}
       &=
    p_{1,i}\sigma_{i}(u_{i})v_{i}
    -
    p_{2,i}u_{i}
    -
    \dot{\alpha}_{i} \nonumber\\
    &=
    p_{1,i}\sigma_{i}(u_{i})
    \frac{
        p_{2,i}u_{i}
        +
        \dot{\alpha}_{i}
        -
        c_{i}\varepsilon_{i}
    }{
        p_{1,i}\sigma_{i}(u_{i})
    }
    -
    p_{2,i}u_{i}
    -
    \dot{\alpha}_{i} ,
\end{align}
which proves \eqref{eq:eps_dynamics_main}. Since $\dot{x}_{i}=u_{i}=\alpha_{i}+\varepsilon_{i}$, we also have
\begin{align}
    \dot{\mathbf{x}}
    =
    -kL\mathbf{x} + \boldsymbol{\varepsilon}.
\end{align}
Applying $\Pi$ and using $\Pi L = L\Pi = L$ yields \eqref{eq:tilde_x_dynamics_main}.
\end{proof}
\begin{theorem}\label{thm:input_consensus_main}
Suppose Assumptions~\ref{ass:connected_graph} and \ref{ass:actuator} hold. Let $ \mathcal{K}
    \subset
    \mathbb{R}^{N} \times \mathcal{U}$ be a compact set of initial conditions. For each $(\mathbf{x}_{0},\mathbf{u}_{0}) \in \mathcal{K}$, define
\begin{align}
    \boldsymbol{\varepsilon}_{0}
    :=
    \mathbf{u}_{0} + kL\mathbf{x}_{0}, \label{eq:eps0_input_main}
\end{align}
and introduce the compact-set radius $R_{\mathcal{K}}
    :=
    \sup_{(\mathbf{x}_{0},\mathbf{u}_{0}) \in \mathcal{K}}
    \left(
        \|\Pi \mathbf{x}_{0}\|^{2}
        +
        \|\boldsymbol{\varepsilon}_{0}\|^{2}
    \right)^{\frac{1}{2}}$. For each agent $i \in \mathcal{V}$, let $\bar{r}_{i}
    :=
    \min\{\overline{u}_{i},-\underline{u}_{i}\}$, and define $r_{i,\mathcal{K}}
    :=
    (2kd_{i}+1)R_{\mathcal{K}}$. Assume that $c_{\min}
    :=
    \min_{1 \le i \le N} c_{i}
    >
    \frac{1}{2k\lambda_{2}(L)}$, and
\begin{align}
    r_{i,\mathcal{K}} < \bar{r}_{i},
    ~~
    \forall i \in \mathcal{V}. \label{eq:interiority_condition_input_main}
\end{align}
Then, for every initial condition $(\mathbf{x}(0),\mathbf{u}(0)) \in \mathcal{K}$, the closed-loop system consisting of \eqref{eq:agent_dynamics}, \eqref{eq:actuator_dynamics}, and \eqref{eq:v_i_main} admits a unique complete solution satisfying the following properties:
\begin{enumerate}
    \item the actuator inputs remain strictly admissible for all time, i.e.,
    \begin{align}
        u_{i}(t) \in \mathcal{U}_{i},
        ~~
        \forall t \ge 0, \ \forall i \in \mathcal{V}; \label{eq:input_admissibility_input_main}
    \end{align}
    \item the actuator-tracking errors satisfy, i.e.,
    \begin{align}
        \varepsilon_{i}(t)
        =
        e^{-c_{i}t}\varepsilon_{i}(0),
        ~~
        \forall t \ge 0, \ \forall i \in \mathcal{V}; \label{eq:eps_explicit_input_main}
    \end{align}
    \item the disagreement and actuator-tracking errors satisfy the estimate, i.e.,
    \begin{align}
        \|\tilde{\mathbf{x}}(t)\|^{2}
        +
        \|\boldsymbol{\varepsilon}(t)\|^{2}
        \le
        e^{-\eta_{\mathrm{c}} t}
        \left(
            \|\tilde{\mathbf{x}}(0)\|^{2}
            +
            \|\boldsymbol{\varepsilon}(0)\|^{2}
        \right),
        ~
        \forall t \ge 0, \label{eq:joint_decay_input_main}
    \end{align}
    where the term $\eta_{\mathrm{c}}$ is defined as 
    \begin{align}
        \eta_{\mathrm{c}}
        :=
        \min
        \left\{
            k\lambda_{2}(L),
            \,
            2c_{\min} - \frac{1}{k\lambda_{2}(L)}
        \right\}
        >
        0; \label{eq:eta_c_input_main}
    \end{align}
    \item there exists a finite constant $x_{\infty} \in \mathbb{R}$ such that
    \begin{align}
        \lim_{t \to \infty} x_{i}(t) = x_{\infty},
        ~~
        \forall i \in \mathcal{V}; \label{eq:consensus_limit_input_main}
    \end{align}
    \item the commanded inputs $v_{i}$ are bounded on $[0,\infty)$, and all closed-loop signals remain bounded.
\end{enumerate}
Consequently, the controller \eqref{eq:v_i_main} solves \Cref{prob:input_constrained_consensus} on the compact admissible set $\mathcal{K}$.
\end{theorem}

\begin{proof}
Fix an arbitrary initial condition $(\mathbf{x}(0),\mathbf{u}(0)) \in \mathcal{K}$. Since $\mathbf{u}(0) \in \mathcal{U}$ and $\sigma_{i}(u_{i}(0))>0$ for all $i$, the closed-loop vector field is well defined at $t=0$, and hence a unique local solution exists on some maximal interval $[0,T_{\max})$.

Consider the Lyapunov function candidate
\begin{align}
    V
    :=
    \frac{1}{2}\|\tilde{\mathbf{x}}\|^{2}
    +
    \frac{1}{2}\|\boldsymbol{\varepsilon}\|^{2}. \label{eq:V_input_main}
\end{align}
Along \eqref{eq:eps_dynamics_main} and \eqref{eq:tilde_x_dynamics_main}, the derivative $\dot{V}$ can be written as
\begin{align}
    \dot{V}
    &=
    \tilde{\mathbf{x}}^{\top}
    \left(
        -kL\tilde{\mathbf{x}} + \Pi\boldsymbol{\varepsilon}
    \right)
    -
    \boldsymbol{\varepsilon}^{\top}C\boldsymbol{\varepsilon} \nonumber\\
    &=
    -k\tilde{\mathbf{x}}^{\top}L\tilde{\mathbf{x}}
    +
    \tilde{\mathbf{x}}^{\top}\Pi\boldsymbol{\varepsilon}
    -
    \boldsymbol{\varepsilon}^{\top}C\boldsymbol{\varepsilon}. \label{eq:Vdot_input_main_1}
\end{align}
Since $\tilde{\mathbf{x}} \in \mathcal{C}^{\perp}$ and $\mathcal{G}$ is connected,
\begin{align}
    \tilde{\mathbf{x}}^{\top}L\tilde{\mathbf{x}}
    \ge
    \lambda_{2}(L)\|\tilde{\mathbf{x}}\|^{2}. \label{eq:lap_lb_input_main}
\end{align}
Also, $\|\Pi\boldsymbol{\varepsilon}\| \le \|\boldsymbol{\varepsilon}\|$, so Young's inequality yields
\begin{align}
    \tilde{\mathbf{x}}^{\top}\Pi\boldsymbol{\varepsilon}
    \le
    \frac{k\lambda_{2}(L)}{2}\|\tilde{\mathbf{x}}\|^{2}
    +
    \frac{1}{2k\lambda_{2}(L)}\|\boldsymbol{\varepsilon}\|^{2}. \label{eq:young_input_main}
\end{align}
Substituting \eqref{eq:lap_lb_input_main} and \eqref{eq:young_input_main} into \eqref{eq:Vdot_input_main_1}, we obtain
\begin{align}
    \dot{V}
    \le
    -\frac{k\lambda_{2}(L)}{2}\|\tilde{\mathbf{x}}\|^{2}
    -
    \left(
        c_{\min} - \frac{1}{2k\lambda_{2}(L)}
    \right)\|\boldsymbol{\varepsilon}\|^{2}. \label{eq:Vdot_input_main_2}
\end{align}
By the gain condition $c_{\min}$, the constant $\eta_{\mathrm{c}}$ defined in \eqref{eq:eta_c_input_main} is strictly positive, and thus $ \dot{V} \le -\eta_{\mathrm{c}}V$. Therefore, $V(t)
    \le
    e^{-\eta_{\mathrm{c}}t}V(0),
    ~
    \forall t \in [0,T_{\max})$. This proves \eqref{eq:joint_decay_input_main} on $[0,T_{\max})$. Indeed,
\begin{align}
    \|\tilde{\mathbf{x}}(t)\|
    \le
    R_{\mathcal{K}},
    ~~
    \|\boldsymbol{\varepsilon}(t)\|
    \le
    R_{\mathcal{K}},
    ~~
    \forall t \in [0,T_{\max}). \label{eq:tildex_eps_bound_input_main}
\end{align}
From \eqref{eq:alpha_i_main}, $|\alpha_{i}(t)|
    =
    k\left|
        \sum_{j=1}^{N} a_{ij}(x_{i}-x_{j})
    \right| = k\left|
        \sum_{j=1}^{N} a_{ij}(\tilde{x}_{i}-\tilde{x}_{j})
    \right|$ reduces to $|\alpha_{i}(t)|\le
    k\sum_{j=1}^{N} a_{ij}
    \left(
        |\tilde{x}_{i}| + |\tilde{x}_{j}|
    \right) \le
    2kd_{i}\|\tilde{\mathbf{x}}(t)\|
    \le
    2kd_{i}R_{\mathcal{K}}$, and hence $|u_{i}(t)|
    \le
    |\alpha_{i}(t)| + |\varepsilon_{i}(t)|
    \le
    (2kd_{i}+1)R_{\mathcal{K}}
    =
    r_{i,\mathcal{K}}~\forall t \in [0,T_{\max})$. By \eqref{eq:interiority_condition_input_main}, this implies
\begin{align}
    u_{i}(t)
    \in
    [-r_{i,\mathcal{K}},r_{i,\mathcal{K}}]
    \subset
    (\underline{u}_{i},\overline{u}_{i}),
    \forall t \in [0,T_{\max}),  \forall i \in \mathcal{V}. \label{eq:u_inner_interval_input_main}
\end{align}
Thus, \eqref{eq:input_admissibility_input_main} holds on $[0,T_{\max})$, and $\sigma_{i}(u_{i}(t))
    \ge
    \underline{\sigma}_{i,\mathcal{K}}
    :=
    \min
    \left\{
        1 - \left(\frac{r_{i,\mathcal{K}}}{\overline{u}_{i}}\right)^{\gamma_{i}},
        \,
        1 - \left(\frac{r_{i,\mathcal{K}}}{-\underline{u}_{i}}\right)^{\gamma_{i}}
    \right\}
    >
    0$ for all $t \in [0,T_{\max})$. Also, from \eqref{eq:alpha_dot_i_main},
    $|\dot{\alpha}_{i}(t)|
    =
    k\left|
        \sum_{j=1}^{N} a_{ij}(u_{i}-u_{j})
    \right|
    \le
    k\sum_{j=1}^{N} a_{ij}
    \left(
        |u_{i}| + |u_{j}|
    \right)
    \le
    2kd_{i}r_{\mathcal{K}}^{\max}$, where $r_{\mathcal{K}}^{\max}
    :=
    \max_{1 \le \ell \le N} r_{\ell,\mathcal{K}}$. Therefore, from \eqref{eq:v_i_main},
\begin{align}
    |v_{i}(t)|
    \le
    \frac{
        p_{2,i}r_{i,\mathcal{K}}
        +
        2kd_{i}r_{\mathcal{K}}^{\max}
        +
        c_{i}R_{\mathcal{K}}
    }{
        p_{1,i}\underline{\sigma}_{i,\mathcal{K}}
    }
    =:
    \bar{v}_{i,\mathcal{K}}, \label{eq:v_bound_input_main}
\end{align}
$\forall t \in [0,T_{\max})$. Hence, the commanded inputs are bounded on $[0,T_{\max})$. To show the completeness, define the network average $\bar{x}(t)
    :=
    \frac{1}{N}\mathbf{1}_{N}^{\top}\mathbf{x}(t)$. Premultiplying $\dot{\mathbf{x}}=-kL\mathbf{x}+\boldsymbol{\varepsilon}$ by $\frac{1}{N}\mathbf{1}_{N}^{\top}$ yields
\begin{align}
    \dot{\bar{x}}(t)
    =
    \frac{1}{N}\mathbf{1}_{N}^{\top}\boldsymbol{\varepsilon}(t). \label{eq:avg_dot_input_main}
\end{align}
Since \eqref{eq:eps_explicit_input_main} holds on $[0,T_{\max})$,
\begin{align}
    |\bar{x}(t)|
    &\le
    |\bar{x}(0)|
    +
    \frac{1}{\sqrt{N}}
    \int_{0}^{t}
    \|\boldsymbol{\varepsilon}(s)\|\,ds \nonumber\\
    &\le
    |\bar{x}(0)|
    +
    \frac{R_{\mathcal{K}}}{\sqrt{N}c_{\min}},
    ~~
    \forall t \in [0,T_{\max}). \label{eq:avg_bound_input_main}
\end{align}
Thus, $\mathbf{x}(t)=\bar{x}(t)\mathbf{1}_{N}+\tilde{\mathbf{x}}(t)$ is bounded on $[0,T_{\max})$ by \eqref{eq:avg_bound_input_main} and \eqref{eq:tildex_eps_bound_input_main}. Together with \eqref{eq:u_inner_interval_input_main} and \eqref{eq:v_bound_input_main}, this shows that all closed-loop signals remain in a compact subset of the domain of the closed-loop vector field on $[0,T_{\max})$. Therefore, by the continuation theorem for ordinary differential equations, finite escape is impossible, and thus $T_{\max}=\infty$.

Now, \eqref{eq:eps_explicit_input_main} follows from \eqref{eq:eps_dynamics_main}, and \eqref{eq:joint_decay_input_main} implies $\tilde{\mathbf{x}}(t)\to \mathbf{0}$ exponentially. Since $\dot{\bar{x}}(\cdot)$ is integrable on $[0,\infty)$ by \eqref{eq:avg_dot_input_main} and \eqref{eq:eps_explicit_input_main}, the average converges to a finite limit
\begin{align}
    x_{\infty}
    :=
    \bar{x}(0)
    +
    \frac{1}{N}\int_{0}^{\infty}\mathbf{1}_{N}^{\top}\boldsymbol{\varepsilon}(s)\,ds. \label{eq:x_infty_input_main}
\end{align}
Hence, \eqref{eq:consensus_limit_input_main} follows from $\mathbf{x}(t)=\bar{x}(t)\mathbf{1}_{N}+\tilde{\mathbf{x}}(t)$.
\end{proof}

\begin{remark}\label{rem:input_semiglobal_main}
\Cref{thm:input_consensus_main} is semiglobal with respect to compact admissible initial sets. The interiority condition \eqref{eq:interiority_condition_input_main} is explicit and checks, \emph{a priori}, that the realized actuator inputs remain in a strict interior subset of their admissible intervals. This avoids any need to assume boundedness of the commanded inputs or positivity of $\sigma_{i}(u_{i})$ along the trajectory.
\end{remark}
\begin{remark}
The interiority condition in \eqref{eq:interiority_condition_input_main}
is conservative because it guarantees admissibility by placing the realized
input $u_i(t)$ inside a symmetric compact subinterval
$[-r_{i,\mathcal{K}},r_{i,\mathcal{K}}]$ contained in the generally
asymmetric admissible set $\mathcal{U}_i$. This choice leads to a compact
semiglobal proof and provides a simple a priori check for strict actuator
admissibility. Less conservative conditions could be obtained by deriving
separate one-sided bounds on the positive and negative excursions of
$u_i(t)$, but such refinements are not pursued here.
\end{remark}
To treat the safe-consensus problem, we strengthen the safe-set description by requiring the existence of a common time-varying admissible core interval.
\begin{assumption}\label{ass:common_safe_core_main}
For all time $t \ge 0$, there exist functions $\underline{\xi},\overline{\xi} : \mathbb{R}_{\ge 0} \to \mathbb{R}$ of class $C^{2}$ and positive constants $\delta_{\xi}$, $\Delta_{\xi}$, $\bar{\xi}_{0}$, $\bar{d}_{\xi}$, and $\bar{s}_{\xi}$ such that $\underline{x}_{i}(t)
    <
    \underline{\xi}(t)
    <
    \overline{\xi}(t)
    <
    \overline{x}_{i}(t)$, for all  $i \in \mathcal{V}$ and $t \ge 0$ with
\begin{align}
    \delta_{\xi}
    \le
    \overline{\xi}(t)-\underline{\xi}(t)
    \le
    \Delta_{\xi},\label{eq:common_safe_core_width_main}
\end{align}
\begin{align}
    |\underline{\xi}(t)|+|\overline{\xi}(t)|
    \le
    \bar{\xi}_{0},
    \label{eq:common_safe_core_bounded}
\end{align}
and
\begin{align}
    |\dot{\underline{\xi}}(t)| + |\dot{\overline{\xi}}(t)|
    \le
    \bar{d}_{\xi},
    ~~
    |\ddot{\underline{\xi}}(t)| + |\ddot{\overline{\xi}}(t)|
    \le
    \bar{s}_{\xi}. \label{eq:common_safe_core_derivatives_main}
\end{align}
\end{assumption}
Let the common admissible core interval (see \Cref{fig:set_based_safe_consensus_strict}) be 
\begin{align}
    \Omega(t)
    :=
    \left(\underline{\xi}(t),\overline{\xi}(t)\right),
    ~~
    t \ge 0,\label{eq:Omega}
\end{align}
and $x^{\star} : \mathbb{R}_{\ge 0} \to \mathbb{R}$ be a prescribed admissible trajectory satisfying $x^{\star}(t) \in \Omega(t) \subseteq \mathcal{X}_{i}(t),$ $\forall i \in \mathcal{V},~
    \forall t \ge 0$. Define the transformed prescribed trajectory
\begin{align}
    z^{\star}(t)
    :=
    \ln
    \left(
        \frac{x^{\star}(t)-\underline{\xi}(t)}
        {\overline{\xi}(t)-x^{\star}(t)}
    \right),
    ~~
    t \ge 0. \label{eq:z_star_main}
\end{align}
Assume that $z^{\star}$ is of class $C^{2}$ and that there exist positive constants $\bar{z}_{\star}$, $\bar{d}_{\star}$, and $\bar{s}_{\star}$ such that
\begin{align}
    |z^{\star}(t)|
    \le
    \bar{z}_{\star},
    ~~
    |\dot{z}^{\star}(t)|
    \le
    \bar{d}_{\star},
    ~~
    |\ddot{z}^{\star}(t)|
    \le
    \bar{s}_{\star},
    ~~
    \forall t \ge 0. \label{eq:z_star_bounds_main}
\end{align}
For each agent $i$, let the barrier coordinate be
\begin{align}
    z_{i}
    :=
    \ln
    \left(
        \frac{x_{i}-\underline{\xi}(t)}
        {\overline{\xi}(t)-x_{i}}
    \right).\label{eq:zidef}
\end{align}

\begin{lemma}\label{lem:barrier_identities_main}
For each fixed $t \ge 0$, the mapping
\begin{align}
    \varphi_{t}(x)
    :=
    \ln
    \left(
        \frac{x-\underline{\xi}(t)}
        {\overline{\xi}(t)-x}
    \right),
    ~~
    x \in \Omega(t),\label{eq:phi_t_main}
\end{align}
is a $C^{1}$-diffeomorphism from $\Omega(t)$ onto $\mathbb{R}$, with inverse
\begin{align}
    \varphi_{t}^{-1}(z)
    =
    \frac{
        \underline{\xi}(t)+\overline{\xi}(t)e^{z}
    }{
        1 + e^{z}
    }. \label{eq:phi_inv_main}
\end{align}
Moreover, along every differentiable trajectory satisfying $x_{i}(t) \in \Omega(t)$, one has $\dot{z}_{i}
    =
    b_{i}(x_{i},t)u_{i} + d_{i}(x_{i},t)$, where
\begin{align}
    b_{i}(x_{i},t)
    :=&
    \frac{
        \overline{\xi}(t)-\underline{\xi}(t)
    }{
        \left(x_{i}-\underline{\xi}(t)\right)
        \left(\overline{\xi}(t)-x_{i}\right)
    }, \label{eq:b_i_main}\\
    d_{i}(x_{i},t)
    :=&
    -
    \frac{\dot{\underline{\xi}}(t)}
    {x_{i}-\underline{\xi}(t)}
    -
    \frac{\dot{\overline{\xi}}(t)}
    {\overline{\xi}(t)-x_{i}}. \label{eq:d_i_main}
\end{align}
\end{lemma}
\begin{proof}
Fix an arbitrary $t\ge 0$. Using \eqref{eq:Omega} and \Cref{ass:common_safe_core_main} gives $\underline{\xi}(t)<\overline{\xi}(t)$. Hence, for every $x\in\Omega(t)$, one has $x-\underline{\xi}(t)>0$ and $\overline{\xi}(t)-x>0$. Therefore, the map \eqref{eq:phi_t_main} is well defined on $\Omega(t)$. Moreover, $\varphi_t(x) = \ln{\left(x-\underline{\xi}(t)\right)} - \ln{\left(\overline{\xi}(t)-x\right)}$ is continuously differentiable with respect to $x$ on $\Omega(t)$, and
\begin{align} 
\frac{\partial \varphi_t}{\partial x}(x) &= \frac{1}{x-\underline{\xi}(t)} + \frac{1}{\overline{\xi}(t)-x} \nonumber\\ 
&= \frac{ \overline{\xi}(t)-\underline{\xi}(t) }{ \left(x-\underline{\xi}(t)\right) \left(\overline{\xi}(t)-x\right) } > 0. \label{eq:phi_x_derivative_main} 
\end{align}
Consequently, $\varphi_t$ is strictly increasing on $\Omega(t)$ and is therefore one-to-one. In addition, using \eqref{eq:phi_t_main}, $\lim_{x\to\underline{\xi}(t)^{+}}\varphi_t(x) = -\infty$, and $\lim_{x\to\overline{\xi}(t)^{-}}\varphi_t(x) = +\infty$. Since $\varphi_t$ is continuous and strictly increasing on $\Omega(t)$, the previous limits imply that $\varphi_t(\Omega(t))=\mathbb{R}$. Hence, $\varphi_t$ is bijective from $\Omega(t)$ onto $\mathbb{R}$.

We next compute its inverse. Let $z=\varphi_t(x)$. Then $e^{z} = \frac{ x-\underline{\xi}(t) }{ \overline{\xi}(t)-x }$ can be rearranged to obtain $e^{z} \left( \overline{\xi}(t)-x \right) = x-\underline{\xi}(t)$ and therefore $x = \frac{ \underline{\xi}(t)+\overline{\xi}(t)e^{z} }{ 1+e^{z} }$. This proves the inverse expression in \eqref{eq:phi_inv_main}. Furthermore, for every $z\in\mathbb{R}$, \eqref{eq:phi_inv_main} gives $\varphi_t^{-1}(z)-\underline{\xi}(t) = \frac{ \left(\overline{\xi}(t)-\underline{\xi}(t)\right)e^{z} }{ 1+e^{z} } > 0$ and $\overline{\xi}(t)-\varphi_t^{-1}(z) = \frac{ \overline{\xi}(t)-\underline{\xi}(t) }{ 1+e^{z} } > 0$. Thus, $\varphi_t^{-1}(z)\in\Omega(t)$ for all $z\in\mathbb{R}$. Also, $\frac{\partial \varphi_t^{-1}}{\partial z}(z) = \left( \overline{\xi}(t)-\underline{\xi}(t) \right) \frac{e^{z}}{\left(1+e^{z}\right)^{2}}$ which is continuous in $z$. Thus, $\varphi_t^{-1}$ is continuously differentiable. Combining this fact with \eqref{eq:phi_x_derivative_main}, $\varphi_t$ is a $C^{1}$-diffeomorphism from $\Omega(t)$ onto $\mathbb{R}$.

It remains to establish the transformed dynamics. Let $x_i(\cdot)$ be any differentiable trajectory satisfying $x_i(t)\in\Omega(t)$, and define $z_i$ as in \eqref{eq:zidef}. Differentiating along the trajectory yields
\begin{align} 
\dot{z}_i &= \frac{ \dot{x}_i-\dot{\underline{\xi}}(t) }{ x_i-\underline{\xi}(t) } - \frac{ \dot{\overline{\xi}}(t)-\dot{x}_i }{ \overline{\xi}(t)-x_i } \nonumber\\
&= \left[ \frac{1}{x_i-\underline{\xi}(t)} + \frac{1}{\overline{\xi}(t)-x_i} \right]\dot{x}_i - \frac{ \dot{\underline{\xi}}(t) }{ x_i-\underline{\xi}(t) } - \frac{ \dot{\overline{\xi}}(t) }{ \overline{\xi}(t)-x_i } \nonumber\\ 
&= \frac{ \overline{\xi}(t)-\underline{\xi}(t) }{ \left(x_i-\underline{\xi}(t)\right) \left(\overline{\xi}(t)-x_i\right) } \dot{x}_i - \frac{ \dot{\underline{\xi}}(t) }{ x_i-\underline{\xi}(t) } - \frac{ \dot{\overline{\xi}}(t) }{ \overline{\xi}(t)-x_i }. \label{eq:zi_dot_derivation_main} 
\end{align}
Using \eqref{eq:agent_dynamics} and comparing \eqref{eq:zi_dot_derivation_main} with \eqref{eq:b_i_main} and \eqref{eq:d_i_main}, we obtain $\dot{z}_i = b_i(x_i,t)u_i+d_i(x_i,t)$. This proves the claimed transformed dynamics and completes the proof.
\end{proof}
Define the transformed synchronization input
\begin{align}
    \beta_{i}
    :=
    \dot{z}^{\star}(t)
    -
    k_{z}\sum_{j=1}^{N} a_{ij}(z_{i}-z_{j})
    -
    \kappa\left(z_{i}-z^{\star}(t)\right),
     \label{eq:beta_i_main}
\end{align}
where $k_{z},\kappa \in \mathbb{R}_{>0}$. Let
\begin{align}
    \alpha_{i}
    :=
    \frac{\beta_{i}-d_{i}(x_{i},t)}
    {b_{i}(x_{i},t)}, \label{eq:alpha_i_safe_main}
\end{align}
and define the actuator-tracking error associated with the safe-consensus virtual input as
\begin{align}
    \varepsilon_i
    :=
    u_i-\alpha_i,
    ~~
    i\in\mathcal{V}.
    \label{eq:eps_safe_def}
\end{align}
The commanded input is then chosen as
\begin{align}
    v_{i}
    =
    \frac{
        p_{2,i}u_{i}
        +
        \dot{\alpha}_{i}
        -
        c_{i}\varepsilon_{i}
    }{
        p_{1,i}\sigma_{i}(u_{i})
    },
    ~~
    c_{i} \in \mathbb{R}_{>0}. \label{eq:v_i_safe_main}
\end{align}
\begin{remark}
The term $\dot{\alpha}_i$ in the commanded signal can be evaluated without introducing any nonlocal information or algebraic dependence on $v_i$. Indeed, direct differentiation of \eqref{eq:alpha_i_safe_main} gives $\dot{\alpha}_i = \frac{\dot{\beta}_i-\dot{d}_i}{b_i} - \frac{\left(\beta_i-d_i\right)\dot{b}_i}{b_i^{2}}$, where, along the closed-loop trajectories, $\dot{b}_i = \frac{\partial b_i}{\partial x_i}u_i + \frac{\partial b_i}{\partial t}$ and $\dot{d}_i = \frac{\partial d_i}{\partial x_i}u_i + \frac{\partial d_i}{\partial t}$. Furthermore, differentiating \eqref{eq:beta_i_main} yields $\dot{\beta}_i = \ddot{z}^{\star}(t) - k_z \sum_{j=1}^{N} a_{ij} \left( \dot{z}_i-\dot{z}_j \right) - \kappa \left( \dot{z}_i-\dot{z}^{\star}(t) \right)$ with $\dot{z}_{\ell} = b_{\ell}(x_{\ell},t)u_{\ell} + d_{\ell}(x_{\ell},t),~\ell\in \{i\}\cup\mathcal{N}_i$. Therefore, agent $i$ can compute $\dot{\alpha}_i$ using its own variables $x_i,u_i$, the neighbor-relative information $\{x_j-x_i,u_j-u_i\}_{j\in\mathcal{N}_i}$, and the known functions $\underline{\xi}(t)$, $\overline{\xi}(t)$, $z^{\star}(t)$ and their derivatives. No derivative of $u_i$ and no commanded signal $v_j$ from neighboring agents is required. Hence, the commanded signal $v_i$ remains distributed in the sense of the information structure in \eqref{eq:controller_state}--\eqref{eq:controller_output}.
\end{remark}
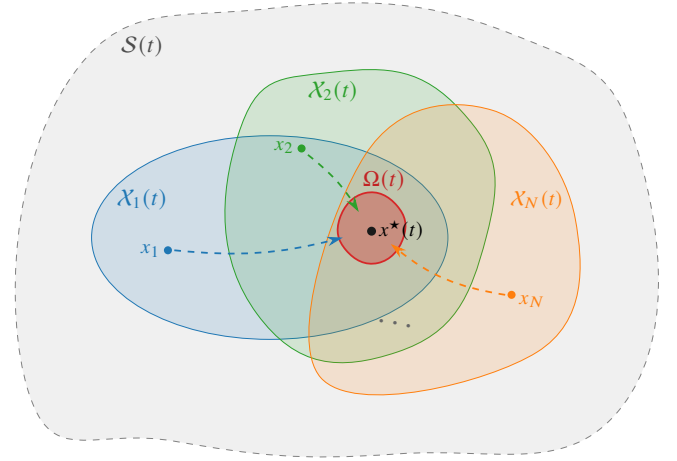
\begin{figure}[ht!]
\centering
\resizebox{\linewidth}{!}{%
\begin{tikzpicture}[
    scale=1.0,
    every node/.style={font=\small},
    >={Stealth[length=2.5mm, width=1.5mm]}
]

\definecolor{setblue}{RGB}{31, 119, 180}
\definecolor{setgreen}{RGB}{44, 160, 44}
\definecolor{setorange}{RGB}{255, 127, 14}
\definecolor{safecore}{RGB}{214, 39, 40}

\filldraw[
    draw=black!50,
    fill=gray!12,
    line width=0 pt,
    dashed
]
plot [smooth cycle, tension=0.8] coordinates
{
    (-4.8, -2.2) (-4.2, 2.6) (-1.0, 3.8) (2.8, 3.4) 
    (4.8, 1.5) (4.5, -2.5) (1.0, -3.2) (-2.5, -3.0)
};
\node[text=black!70, font=\normalsize\bfseries] at (-3.0, 3.2) {$\mathcal{S}(t)$};

\filldraw[
    draw=setblue,
    fill=setblue,
    fill opacity=0.15,
    line width=0.2pt
]
(-1.0, 0.2) ellipse (2.8 and 1.6);
\node[setblue, font=\normalsize\bfseries] at (-3, 0.8) {$\mathcal{X}_1(t)$};

\filldraw[
    draw=setgreen,
    fill=setgreen,
    fill opacity=0.15,
    line width=0.2pt
]
plot [smooth cycle, tension=0.9] coordinates
{
    (-1.2, -1.0) (-1.5, 1.8) (0.0, 2.8) (2.0, 2.2) (2.4, 0.2) (1.0, -1.6)
};
\node[setgreen, font=\normalsize\bfseries] at (0, 2.5) {$\mathcal{X}_2(t)$};

\filldraw[
    draw=setorange,
    fill=setorange,
    fill opacity=0.15,
    line width=0.2pt
]
plot [smooth cycle, tension=0.7] coordinates
{
    (-0.2, -2.0) (3.2, -1.6) (3.8, 0.5) (2.5, 2.2) (0.5, 1.5)
};
\node[setorange] at (3.2, 0.8) {$\mathcal{X}_N(t)$};

\node[text=black!60, font=\Large, rotate=-10] at (1.0, -1.2) {$\cdots$};

\filldraw[
    draw=safecore,
    fill=safecore,
    fill opacity=0.35,
    line width=0.8pt
]
plot [smooth cycle, tension=0.8] coordinates
{
    (0.2, 0.0) (0.1, 0.5) (0.5, 0.9) (1.0, 0.7) (1.1, 0.2) (0.7, -0.2)
};
\node[text=safecore!90!black, font=\normalsize\bfseries] at (0.8, 1.1) {$\Omega(t)$};

\fill[black!90] (0.6, 0.3) circle (2pt);
\node[below right] at (0.6, 0.6) {$x^\star(t)$};

\fill[setblue] (-2.6, 0.0) circle (1.8pt);
\node[left, setblue] at (-2.6, 0.0) {$x_1$};

\fill[setgreen] (-0.5, 1.6) circle (1.8pt);
\node[above left, setgreen] at (-0.5, 1.4) {$x_2$};

\fill[setorange] (2.8, -0.7) circle (1.8pt);
\node[right, setorange] at (2.8, -0.8) {$x_N$};

\draw[->, setblue, thick, dashed] (-2.5, 0.0) to[bend right=10] (0.15, 0.2);
\draw[->, setgreen, thick, dashed] (-0.4, 1.5) to[bend left=10] (0.4, 0.6);
\draw[->, setorange, thick, dashed] (2.7, -0.7) to[bend left=15] (0.9, 0.1);

\end{tikzpicture}%
}
\caption{Set-based illustration of the proposed safe-consensus framework. Each agent $i$ operates within its individual admissible set $\mathcal{X}_i(t)$, all of which reside within the broader state space $\mathcal{S}(t)$. The common safe core $\Omega(t)$ is assumed to lie strictly inside the
mutual intersection $\bigcap_i\mathcal{X}_i(t)$. Coordination ensures the network converges to a synchronized reference $x^\star(t)\in\Omega(t)$.}
\label{fig:set_based_safe_consensus_strict}
\end{figure}
\begin{proposition}\label{prop:reduced_safe_dynamics}
Consider the closed-loop system consisting of \eqref{eq:agent_dynamics}, \eqref{eq:actuator_dynamics}, and \eqref{eq:v_i_safe_main}. Assume that the corresponding solution is well defined on an interval on which $x_{i}(t) \in \Omega(t)$, $u_{i}(t) \in \mathcal{U}_{i}$, $\forall i \in \mathcal{V}$. Let $\mathbf{z}
    :=
    \mathrm{col}\{z_{1},\dots,z_{N}\}$, $\boldsymbol{\zeta}
    :=
    \mathbf{z}-z^{\star}(t)\mathbf{1}_{N}$, and $B(\mathbf{x},t)
    :=
    \mathrm{diag}\{b_{1}(x_{1},t),\dots,b_{N}(x_{N},t)\}$. Then one has
\begin{align}
    \dot{\varepsilon}_{i}
    =
    -c_{i}\varepsilon_{i},
    ~~
    i \in \mathcal{V}, \label{eq:eps_dynamics_safe_main}
\end{align}
and
\begin{align}
    \dot{\boldsymbol{\zeta}}
    =
    -\left(k_{z}L+\kappa I_{N}\right)\boldsymbol{\zeta}
    +
    B(\mathbf{x},t)\boldsymbol{\varepsilon}. \label{eq:zeta_dynamics_safe_main}
\end{align}
\end{proposition}
\begin{proof}
From \Cref{lem:barrier_identities_main}, $\dot{z}_{i}
    =
    b_{i}(x_{i},t)u_{i} + d_{i}(x_{i},t)$. Using $u_{i}=\alpha_{i}+\varepsilon_{i}$ and \eqref{eq:alpha_i_safe_main}, we obtain
\begin{align}
    \dot{z}_{i}
    =&
    b_{i}(x_{i},t)\alpha_{i}
    +
    d_{i}(x_{i},t)
    +
    b_{i}(x_{i},t)\varepsilon_{i}\nonumber\\
    =&
    \beta_{i} + b_{i}(x_{i},t)\varepsilon_{i}.
\end{align}
Substituting \eqref{eq:beta_i_main} above gives
\begin{align}
    \dot{z}_{i}-\dot{z}^{\star}(t)
    =&
    -
    k_{z}\sum_{j=1}^{N}a_{ij}
    \left[
        \left(z_{i}-z^{\star}(t)\right)
        -
        \left(z_{j}-z^{\star}(t)\right)
    \right]\nonumber\\
    &-
    \kappa\left(z_{i}-z^{\star}(t)\right)
    +
    b_{i}(x_{i},t)\varepsilon_{i},
\end{align}
which yields \eqref{eq:zeta_dynamics_safe_main} in stacked form. The proof of \eqref{eq:eps_dynamics_safe_main} is identical to that of \Cref{prop:reduced_input_dynamics}.
\end{proof}

\begin{theorem}\label{thm:safe_consensus_main}
Suppose \Cref{ass:connected_graph,ass:actuator,ass:common_safe_core_main} hold. Let $\mathcal{K}
    \subset
    \Omega(0)^{N} \times \mathcal{U}$ be a compact set of initial conditions. For each $(\mathbf{x}_{0},\mathbf{u}_{0}) \in \mathcal{K}$, define
\begin{align}
    \mathbf{z}_{0}
    :=
    \mathrm{col}
    \left\{
        \ln
        \left(
            \frac{x_{1,0}-\underline{\xi}(0)}
            {\overline{\xi}(0)-x_{1,0}}
        \right),
        \dots,
        \ln
        \left(
            \frac{x_{N,0}-\underline{\xi}(0)}
            {\overline{\xi}(0)-x_{N,0}}
        \right)
    \right\},
\end{align}
\begin{align}
    \boldsymbol{\varepsilon}_{0}
    :=
    \mathbf{u}_{0} - \boldsymbol{\alpha}(\mathbf{x}_{0},0),
\end{align}
and
\begin{align}
    S_{\mathcal{K}}
    :=
    \sup_{(\mathbf{x}_{0},\mathbf{u}_{0}) \in \mathcal{K}}
    \left(
        \|\mathbf{z}_{0}-z^{\star}(0)\mathbf{1}_{N}\|^{2}
        +
        \|\boldsymbol{\varepsilon}_{0}\|^{2}
    \right)^{\frac{1}{2}}.
\end{align}
Let $M_{\mathcal{K}}
    :=
    \bar{z}_{\star} + S_{\mathcal{K}},
    ~
    \bar{b}_{\mathcal{K}}
    :=
    \frac{2(1+e^{M_{\mathcal{K}}})}{\delta_{\xi}},
    ~
    \bar{d}_{\mathcal{K}}
    :=
    \frac{2\bar{d}_{\xi}(1+e^{M_{\mathcal{K}}})}{\delta_{\xi}}$,
and, for each $i \in \mathcal{V}$, $\bar{\beta}_{i,\mathcal{K}}
    :=
    \bar{d}_{\star}
    +
    (2k_{z}d_{i}+\kappa)S_{\mathcal{K}}$, $\bar{\alpha}_{i,\mathcal{K}}
    :=
    \frac{\Delta_{\xi}}{4}
    \left(
        \bar{\beta}_{i,\mathcal{K}} + \bar{d}_{\mathcal{K}}
    \right),
    ~
    r_{i,\mathcal{K}}
    :=
    \bar{\alpha}_{i,\mathcal{K}} + S_{\mathcal{K}},
    ~
    \bar{r}_{i}
    :=
    \min\{\overline{u}_{i},-\underline{u}_{i}\}$.
Assume that
\begin{align}
    c_{\min}
    :=&
    \min_{1 \le i \le N} c_{i}
    >
    \frac{\bar{b}_{\mathcal{K}}^{2}}{2\kappa}, \label{eq:gain_condition_safe_main}\\
    r_{i,\mathcal{K}} <& \bar{r}_{i},
    ~~
    \forall i \in \mathcal{V}. \label{eq:interiority_condition_safe_main}
\end{align}
Then, for every initial condition $(\mathbf{x}(0),\mathbf{u}(0)) \in \mathcal{K}$, the closed-loop system consisting of \eqref{eq:agent_dynamics}, \eqref{eq:actuator_dynamics}, and \eqref{eq:v_i_safe_main} admits a unique complete solution satisfying the following properties:
\begin{enumerate}
    \item the output remains safe and the actuator inputs remain strictly admissible for all time, i.e.,
    \begin{align}
        x_{i}(t) \in \Omega(t) \subseteq \mathcal{X}_{i}(t),
        ~
        u_{i}(t) \in \mathcal{U}_{i},
        ~
        \forall t \ge 0,
        ~
        \forall i \in \mathcal{V}; \label{eq:safety_and_admissibility_safe_main}
    \end{align}
    \item the actuator-tracking errors satisfy
    \begin{align}
        \varepsilon_{i}(t)
        =
        e^{-c_{i}t}\varepsilon_{i}(0),
        ~~
        \forall t \ge 0,
        ~~
        \forall i \in \mathcal{V}; \label{eq:eps_explicit_safe_main}
    \end{align}
    \item the transformed synchronization error satisfies
    \begin{align}
        \|\boldsymbol{\zeta}(t)\|^{2}
        +
        \|\boldsymbol{\varepsilon}(t)\|^{2}
        \le
        e^{-\eta_{\mathrm{s}} t}
        \left(
            \|\boldsymbol{\zeta}(0)\|^{2}
            +
            \|\boldsymbol{\varepsilon}(0)\|^{2}
        \right),
        ~~
        \forall t \ge 0, \label{eq:joint_decay_safe_main}
    \end{align}
    where the term $\eta_{\mathrm{s}}$ is defined as 
    \begin{align}
        \eta_{\mathrm{s}}
        :=
        \min
        \left\{
            \kappa,
            \,
            2c_{\min} - \frac{\bar{b}_{\mathcal{K}}^{2}}{\kappa}
        \right\}
        >
        0; \label{eq:eta_s_safe_main}
    \end{align}
    \item the agent states asymptotically synchronize to the prescribed admissible trajectory $x^{\star}(t)$, i.e.,
    \begin{align}
        \lim_{t \to \infty}
        \left(
            x_{i}(t)-x^{\star}(t)
        \right)
        =
        0,
        ~~
        \forall i \in \mathcal{V}; \label{eq:x_to_xstar_safe_main}
    \end{align}
    and specifically,
    \begin{align}
        \lim_{t \to \infty}
        |x_{i}(t)-x_{j}(t)|
        =
        0,
        ~~
        \forall i,j \in \mathcal{V}; \label{eq:consensus_safe_main}
    \end{align}
    \item the commanded inputs $v_{i}$ are bounded on $[0,\infty)$, and all closed-loop signals remain bounded.
\end{enumerate}
Consequently, the controller \eqref{eq:v_i_safe_main} solves \Cref{prob:safe_consensus} on the compact admissible set $\mathcal{K}$, and in fact establishes the stronger objective of synchronization to the prescribed admissible trajectory $x^{\star}(t)$.
\end{theorem}
\begin{proof}
Fix an arbitrary initial condition $(\mathbf{x}(0),\mathbf{u}(0)) \in \mathcal{K}$. Since $\mathcal{K} \subset \Omega(0)^{N} \times \mathcal{U}$, the barrier coordinates are well defined at $t=0$, and the closed-loop vector field is locally well defined. Let $[0,T_{\max})$ denote the corresponding maximal interval of existence. Define $T^{\star}
    :=
    \sup
    \left\{
        T \in (0,T_{\max}] :
        x_{i}(t) \in \Omega(t)
        \text{ for all } t \in [0,T),
        ~
        \forall i \in \mathcal{V}
    \right\}$. Clearly, $T^{\star}>0$. We first analyze the dynamics on $[0,T^{\star})$, where \Cref{prop:reduced_safe_dynamics} applies.

Consider the Lyapunov function candidate
\begin{align}
    W
    :=
    \frac{1}{2}\|\boldsymbol{\zeta}\|^{2}
    +
    \frac{1}{2}\|\boldsymbol{\varepsilon}\|^{2}. \label{eq:W_safe_main}
\end{align}
On $[0,T^{\star})$, \Cref{prop:reduced_safe_dynamics} gives
\begin{align}
    \dot{W}
    &=
    \boldsymbol{\zeta}^{\top}
    \left(
        -\left(k_{z}L+\kappa I_{N}\right)\boldsymbol{\zeta}
        +
        B(\mathbf{x},t)\boldsymbol{\varepsilon}
    \right)
    -
    \boldsymbol{\varepsilon}^{\top}C\boldsymbol{\varepsilon} \nonumber\\
    &=
    -k_{z}\boldsymbol{\zeta}^{\top}L\boldsymbol{\zeta}
    -
    \kappa\|\boldsymbol{\zeta}\|^{2}
    +
    \boldsymbol{\zeta}^{\top}B(\mathbf{x},t)\boldsymbol{\varepsilon}
    -
    \boldsymbol{\varepsilon}^{\top}C\boldsymbol{\varepsilon}. \label{eq:Wdot_safe_main_1}
\end{align}
We first derive a trajectory-independent bound on $B(\mathbf{x},t)$ over the slab $|z_{i}| \le M_{\mathcal{K}}$. From \eqref{eq:phi_inv_main}, if $|z_{i}| \le M_{\mathcal{K}}$, then
\begin{align}
    x_{i}-\underline{\xi}(t)
    =
    \frac{
        \left(\overline{\xi}(t)-\underline{\xi}(t)\right)e^{z_{i}}
    }{
        1 + e^{z_{i}}
    },
    ~~
    \overline{\xi}(t)-x_{i}
    =
    \frac{
        \overline{\xi}(t)-\underline{\xi}(t)
    }{
        1 + e^{z_{i}}
    }.
\end{align}
Substituting these expressions into \eqref{eq:b_i_main} yields
\begin{align}
    b_{i}(x_{i},t)
    =
    \frac{
        e^{-z_{i}} + 2 + e^{z_{i}}
    }{
        \overline{\xi}(t)-\underline{\xi}(t)
    }
    \le
    \frac{2(1+e^{M_{\mathcal{K}}})}{\delta_{\xi}}
    =
    \bar{b}_{\mathcal{K}}. \label{eq:b_bound_safe_main}
\end{align}
Hence, if $|z_{i}(t)| \le M_{\mathcal{K}}$ for all $i$, then $\|B(\mathbf{x}(t),t)\| \le \bar{b}_{\mathcal{K}}$. Now, on any subinterval of $[0,T^{\star})$ where $|z_{i}(t)| \le M_{\mathcal{K}}$ for all $i$, one has from \eqref{eq:Wdot_safe_main_1}
\begin{align}
    \dot{W}
    \le
    -\kappa\|\boldsymbol{\zeta}\|^{2}
    +
    \bar{b}_{\mathcal{K}}\|\boldsymbol{\zeta}\|\,\|\boldsymbol{\varepsilon}\|
    -
    c_{\min}\|\boldsymbol{\varepsilon}\|^{2}.
\end{align}
Applying Young's inequality,
\begin{align}
    \bar{b}_{\mathcal{K}}\|\boldsymbol{\zeta}\|\,\|\boldsymbol{\varepsilon}\|
    \le
    \frac{\kappa}{2}\|\boldsymbol{\zeta}\|^{2}
    +
    \frac{\bar{b}_{\mathcal{K}}^{2}}{2\kappa}\|\boldsymbol{\varepsilon}\|^{2},
\end{align}
so we obtain
\begin{align}
    \dot{W}
    \le
    -\frac{\kappa}{2}\|\boldsymbol{\zeta}\|^{2}
    -
    \left(
        c_{\min} - \frac{\bar{b}_{\mathcal{K}}^{2}}{2\kappa}
    \right)\|\boldsymbol{\varepsilon}\|^{2}.
\end{align}
By \eqref{eq:gain_condition_safe_main}, the constant $\eta_{\mathrm{s}}$ defined in \eqref{eq:eta_s_safe_main} is strictly positive, and therefore
\begin{align}
    \dot{W}
    \le
    -\eta_{\mathrm{s}}W. \label{eq:Wdot_safe_main_4}
\end{align}
Since $W(0) \le \frac{1}{2}S_{\mathcal{K}}^{2}$, a standard continuation argument shows that \eqref{eq:Wdot_safe_main_4} holds on the whole interval $[0,T^{\star})$. Indeed, if there existed a first time $t_{1} \in (0,T^{\star})$ such that $|z_{i}(t_{1})| > M_{\mathcal{K}}$ for some $i$, then by continuity there would be a first time at which $|z_{i}|=M_{\mathcal{K}}$. Integrating \eqref{eq:Wdot_safe_main_4} up to that time would give
\begin{align}
    W(t)
    \le
    e^{-\eta_{\mathrm{s}}t}W(0)
    \le
    W(0)
    \le
    \frac{1}{2}S_{\mathcal{K}}^{2},
    ~~
    \forall t \in [0,T^{\star}), \label{eq:W_exp_safe_main}
\end{align}
and hence
\begin{align}
    \|\boldsymbol{\zeta}(t)\|
    \le
    S_{\mathcal{K}},
    ~~
    \forall t \in [0,T^{\star}). \label{eq:zeta_bound_safe_main}
\end{align}
Therefore, using \eqref{eq:z_star_bounds_main}, $\forall i \in \mathcal{V}$,
\begin{align}
    |z_{i}(t)|
    \le
    |z^{\star}(t)| + \|\boldsymbol{\zeta}(t)\|
    \le
    \bar{z}_{\star}+S_{\mathcal{K}}
    =
    M_{\mathcal{K}},~
    \forall t \in [0,T^{\star}),
    \label{eq:z_bound_safe_main}
\end{align}
which is a contradiction. Thus \eqref{eq:W_exp_safe_main} and \eqref{eq:z_bound_safe_main} hold on all of $[0,T^{\star})$, and \eqref{eq:joint_decay_safe_main} follows immediately.

We next use \eqref{eq:z_bound_safe_main} to prove strict forward invariance of the safe set. From \eqref{eq:phi_inv_main}, and using \eqref{eq:common_safe_core_width_main} and \eqref{eq:z_bound_safe_main}, we obtain
\begin{align}
    x_{i}(t)-\underline{\xi}(t)
    \ge&
    \frac{\delta_{\xi}}{1+e^{M_{\mathcal{K}}}}
    =:
    m_{\mathcal{K}}
    >
    0,\\
    \overline{\xi}(t)-x_{i}(t)
    \ge&
    \frac{\delta_{\xi}}{1+e^{M_{\mathcal{K}}}}
    =
    m_{\mathcal{K}}
    >
    0,
\end{align}
for all $t \in [0,T^{\star})$ and all $i \in \mathcal{V}$. Hence, $\forall t \in [0,T^{\star})$, $\forall i \in \mathcal{V}$,
\begin{align}
    x_{i}(t)
    \in
    \left[
        \underline{\xi}(t)+m_{\mathcal{K}},
        \,
        \overline{\xi}(t)-m_{\mathcal{K}}
    \right]
    \subset
    \Omega(t). \label{eq:compact_strip_safe_main}
\end{align}
Thus, the solution cannot approach the boundary of $\Omega(t)$ in finite time, and therefore $T^{\star}=T_{\max}$.

We now bound the realized actuator inputs. First, from \eqref{eq:beta_i_main}, \eqref{eq:z_star_bounds_main}, and \eqref{eq:zeta_bound_safe_main},
\begin{align}
    |\beta_{i}(t)|
    \le&
    \bar{d}_{\star}
    +
    2k_{z}d_{i}\|\boldsymbol{\zeta}(t)\|
    +
    \kappa\|\boldsymbol{\zeta}(t)\|\nonumber\\
    \le&
    \bar{d}_{\star}
    +
    (2k_{z}d_{i}+\kappa)S_{\mathcal{K}}
    =
    \bar{\beta}_{i,\mathcal{K}}.
\end{align}
Also, from \eqref{eq:d_i_main}, \eqref{eq:common_safe_core_derivatives_main}, and \eqref{eq:compact_strip_safe_main},
\begin{align}
    |d_{i}(x_{i}(t),t)|
    \le&
    \frac{\bar{d}_{\xi}}{x_{i}(t)-\underline{\xi}(t)}
    +
    \frac{\bar{d}_{\xi}}{\overline{\xi}(t)-x_{i}(t)}\nonumber\\
    \le&
    \frac{2\bar{d}_{\xi}(1+e^{M_{\mathcal{K}}})}{\delta_{\xi}}
    =
    \bar{d}_{\mathcal{K}}.
\end{align}
Moreover, from \eqref{eq:phi_inv_main} and \eqref{eq:b_i_main},
\begin{align}
    \frac{1}{b_{i}(x_{i}(t),t)}
    =
    \frac{
        \left(\overline{\xi}(t)-\underline{\xi}(t)\right)e^{z_{i}(t)}
    }{
        \left(1+e^{z_{i}(t)}\right)^{2}
    }
    \le
    \frac{\Delta_{\xi}}{4}.
\end{align}
Hence, from \eqref{eq:alpha_i_safe_main},
\begin{align}
    |\alpha_{i}(t)|
    \le
    \frac{\Delta_{\xi}}{4}
    \left(
        \bar{\beta}_{i,\mathcal{K}} + \bar{d}_{\mathcal{K}}
    \right)
    =
    \bar{\alpha}_{i,\mathcal{K}}.
\end{align}
Since \eqref{eq:joint_decay_safe_main} implies $\|\boldsymbol{\varepsilon}(t)\| \le S_{\mathcal{K}}$, we conclude that
\begin{align}
    |u_{i}(t)|
    \le
    |\alpha_{i}(t)| + |\varepsilon_{i}(t)|
    \le
    \bar{\alpha}_{i,\mathcal{K}} + S_{\mathcal{K}}
    =
    r_{i,\mathcal{K}},
    ~
    \forall t \in [0,T_{\max}).
\end{align}
From \eqref{eq:interiority_condition_safe_main},
\begin{align}
    u_{i}(t)
    \in
    [-r_{i,\mathcal{K}},r_{i,\mathcal{K}}]
    \subset
    (\underline{u}_{i},\overline{u}_{i}),
    ~
    \forall t \in [0,T_{\max}),
    ~
    \forall i \in \mathcal{V}. \label{eq:u_inner_interval_safe_main}
\end{align}
Thus \eqref{eq:safety_and_admissibility_safe_main} holds on $[0,T_{\max})$, and
\begin{align}
    \sigma_{i}(u_{i}(t))
    \ge
    \underline{\sigma}_{i,\mathcal{K}}
    :=
    \min
    \left\{
        1 - \left(\frac{r_{i,\mathcal{K}}}{\overline{u}_{i}}\right)^{\gamma_{i}},
        \,
        1 - \left(\frac{r_{i,\mathcal{K}}}{-\underline{u}_{i}}\right)^{\gamma_{i}}
    \right\}
    >
    0. \label{eq:sigma_lower_safe_main}
\end{align}
It remains to establish boundedness of the commanded inputs and completeness. The quantities $\underline{\xi}(t)$, $\overline{\xi}(t)$, and their
derivatives up to second order are bounded by
\eqref{eq:common_safe_core_bounded} and
\eqref{eq:common_safe_core_derivatives_main}. Moreover, $z^{\star}(t)$ and its derivatives up to second order are bounded by \eqref{eq:z_star_bounds_main}. The estimate \eqref{eq:z_bound_safe_main} shows that $\mathbf{z}(t)$ remains in a compact slab, and therefore \eqref{eq:phi_inv_main} implies that $\mathbf{x}(t)$ evolves in the compact strip \eqref{eq:compact_strip_safe_main}. Since $\beta_{i}$ depends smoothly on $\mathbf{z}$, $z^{\star}$, and $\dot{z}^{\star}$, and since $\alpha_{i}$ is a smooth function of these quantities on the compact strip, it follows that $\dot{\alpha}_{i}$ is bounded on $[0,T_{\max})$. Combined with \eqref{eq:u_inner_interval_safe_main}, \eqref{eq:sigma_lower_safe_main}, and \eqref{eq:v_i_safe_main}, this yields boundedness of $v_{i}$ on $[0,T_{\max})$.

Hence, all closed-loop signals remain in a compact subset of the domain of the closed-loop vector field, so finite escape is impossible. Therefore, $T_{\max}=\infty$. The decay estimate \eqref{eq:joint_decay_safe_main} implies $\boldsymbol{\zeta}(t)\to \mathbf{0}$ and $\boldsymbol{\varepsilon}(t)\to \mathbf{0}$. Thus
\begin{align}
    z_{i}(t)-z^{\star}(t)
    \to
    0,
    ~~
    \forall i \in \mathcal{V}.
\end{align}
Furthermore, the inverse barrier map \eqref{eq:phi_inv_main} has derivative
\begin{align}
    \frac{\partial \varphi_{t}^{-1}}{\partial z}(z)
    =
    \frac{
        \left(\overline{\xi}(t)-\underline{\xi}(t)\right)e^{z}
    }{
        \left(1+e^{z}\right)^{2}
    }
    \le
    \frac{\Delta_{\xi}}{4},
\end{align}
for all $z \in \mathbb{R}$ and all $t \ge 0$. Since $x_{i}(t)=\varphi_{t}^{-1}(z_{i}(t))$ and $x^{\star}(t)=\varphi_{t}^{-1}(z^{\star}(t))$, the mean-value theorem gives
\begin{align}
    |x_{i}(t)-x^{\star}(t)|
    \le
    \frac{\Delta_{\xi}}{4}
    |z_{i}(t)-z^{\star}(t)|,
    ~~
    \forall t \ge 0.
\end{align}
Therefore, \eqref{eq:x_to_xstar_safe_main} holds, and \eqref{eq:consensus_safe_main} follows immediately.
\end{proof}

\begin{remark}\label{rem:safe_semiglobal_main}
\Cref{thm:safe_consensus_main} is semiglobal with respect to compact admissible initial sets. The conditions \eqref{eq:gain_condition_safe_main} and \eqref{eq:interiority_condition_safe_main} depend only on the design gains, the graph, the actuator bounds, the regularity constants of the common safe core, the boundedness constants of the prescribed transformed trajectory, and the chosen compact set of initial conditions. No closed-loop trajectory-dependent quantity is assumed a priori. As in \Cref{thm:input_consensus_main}, condition \eqref{eq:interiority_condition_input_main} uses a symmetric inner bound contained in the asymmetric actuator interval. Hence, it is sufficient but not necessary for strict actuator admissibility.
\end{remark}

\begin{remark}\label{rem:pinning_main}
The pinning term $-\kappa\left(z_{i}-z^{\star}(t)\right)$ in \eqref{eq:beta_i_main} removes the neutral consensus mode in the transformed coordinates, while the feedforward term $\dot{z}^{\star}(t)$ accounts for the motion of the prescribed admissible trajectory in the same coordinates. Consequently, the closed-loop system achieves a stronger objective than plain agreement, namely, synchronization to the prescribed admissible trajectory $x^{\star}(t)$.
\end{remark}
\begin{remark}
The proposed safe consensus framework differs from a direct funnel-control or prescribed-performance construction \cite{chen2020tcns_prescribed_performance,10794772} in three main ways: \emph{(i)} In their standard formulation, the associated transformation is therefore an error transformation, and the main objective is to keep the error inside a prescribed transient envelope. In the proposed design, the transformation used here is not applied to a consensus error. It is applied directly to the agent output relative to the admissible core interval \eqref{eq:zidef} so that boundedness of $z_i$ is equivalent to strict interiority of the physical output constraint $x_i(t) \in \Omega(t)$. \emph{(ii)} The proposed framework separates the admissible trajectory selection from the consensus mechanism by assigning the desired group behavior via a designer-selected admissible trajectory $x^{\star}(t)\in\Omega(t)$ and its transformed representation \eqref{eq:z_star_main}, so the closed-loop network synchronizes to a specified safe trajectory rather than to an unspecified agreement value. \emph{(iii)} Unlike funnel-control approaches that primarily encode output-error performance and may require a separate treatment of actuator feasibility, the present framework provides a unified distributed mechanism for strict asymmetric input admissibility, forward invariance of the output safe set, and synchronization to a prescribed safe trajectory.
\end{remark}
\begin{corollary}\label{cor:fixed_transformed_reference_main}
Suppose the hypotheses of \Cref{thm:safe_consensus_main} hold. If the prescribed admissible trajectory is generated by a constant transformed
coordinate $z_c^{\star}\in\mathbb{R}$ as
\begin{align}
    x^{\star}(t)
    =
    \varphi_t^{-1}(z_c^{\star})
    =
    \frac{
        \underline{\xi}(t)+\overline{\xi}(t)e^{z_c^{\star}}
    }{
        1+e^{z_c^{\star}}
    },
    ~~
    t\ge 0,
    \label{eq:fixed_transformed_reference}
\end{align}
then $z^{\star}(t)\equiv z_c^{\star}$ and
$\dot{z}^{\star}(t)\equiv 0$. In this case, the transformed
synchronization input reduces to
\begin{align}
    \beta_i
    =
    -
    k_z\sum_{j=1}^{N}a_{ij}(z_i-z_j)
    -
    \kappa(z_i-z_c^{\star}),
    \label{eq:beta_fixed_transformed_reference}
\end{align}
and the closed-loop system satisfies
\begin{align}
    \lim_{t\to\infty}
    \left(
        x_i(t)-x^{\star}(t)
    \right)
    =
    0,
    ~~
    \forall i\in\mathcal{V}.
\end{align}
If, in addition, the common safe core interval is time invariant, that is,
$\underline{\xi}(t)\equiv\underline{\xi}$ and
$\overline{\xi}(t)\equiv\overline{\xi}$, then
\begin{align}
    x^{\star}(t)
    \equiv
    x_c^{\star}
    :=
    \frac{
        \underline{\xi}+\overline{\xi}e^{z_c^{\star}}
    }{
        1+e^{z_c^{\star}}
    },
\end{align}
and hence
\begin{align}
    \lim_{t\to\infty}x_i(t)
    =
    x_c^{\star},
    ~~
    \forall i\in\mathcal{V}.
\end{align}
\end{corollary}

\begin{proof}
If \eqref{eq:fixed_transformed_reference} holds, then applying
$\varphi_t$ to both sides gives $z^{\star}(t)\equiv z_c^{\star}$, and
therefore $\dot{z}^{\star}(t)\equiv 0$. Hence the transformed
synchronization input reduces to
\eqref{eq:beta_fixed_transformed_reference}. The convergence statement
follows directly from \Cref{thm:safe_consensus_main}. If
$\underline{\xi}$ and $\overline{\xi}$ are constant, then
\eqref{eq:fixed_transformed_reference} is constant, which gives the final
claim.
\end{proof}

\section{Simulation Results}\label{sec:simulations} 
\begin{figure*}[ht!]
    \centering
    \begin{subfigure}[t]{.325\linewidth}
        \centering
        \includegraphics[width=\linewidth]{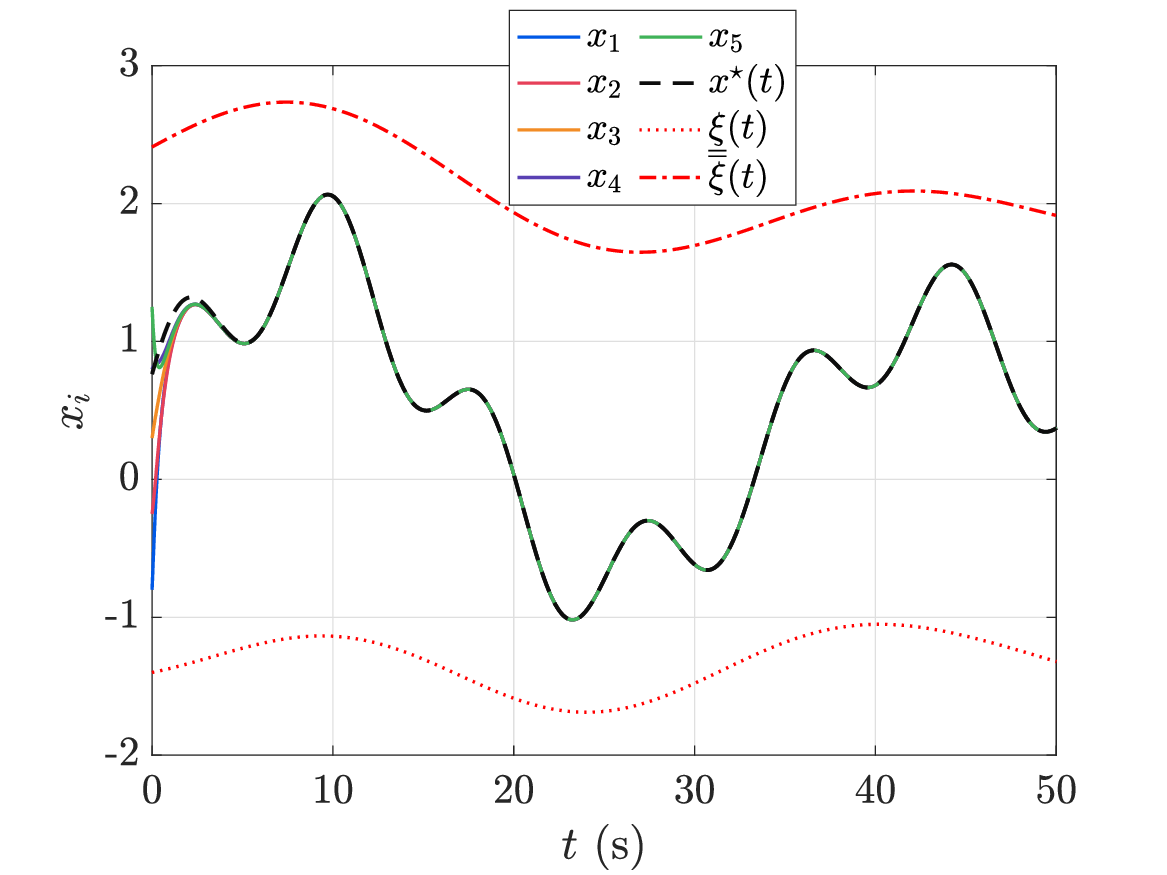}
    \caption{State trajectories $x_{i}(t)$.}
    \label{fig:xi_mf}
    \end{subfigure}
    \begin{subfigure}[t]{.325\linewidth}
        \centering
        \includegraphics[width=\linewidth]{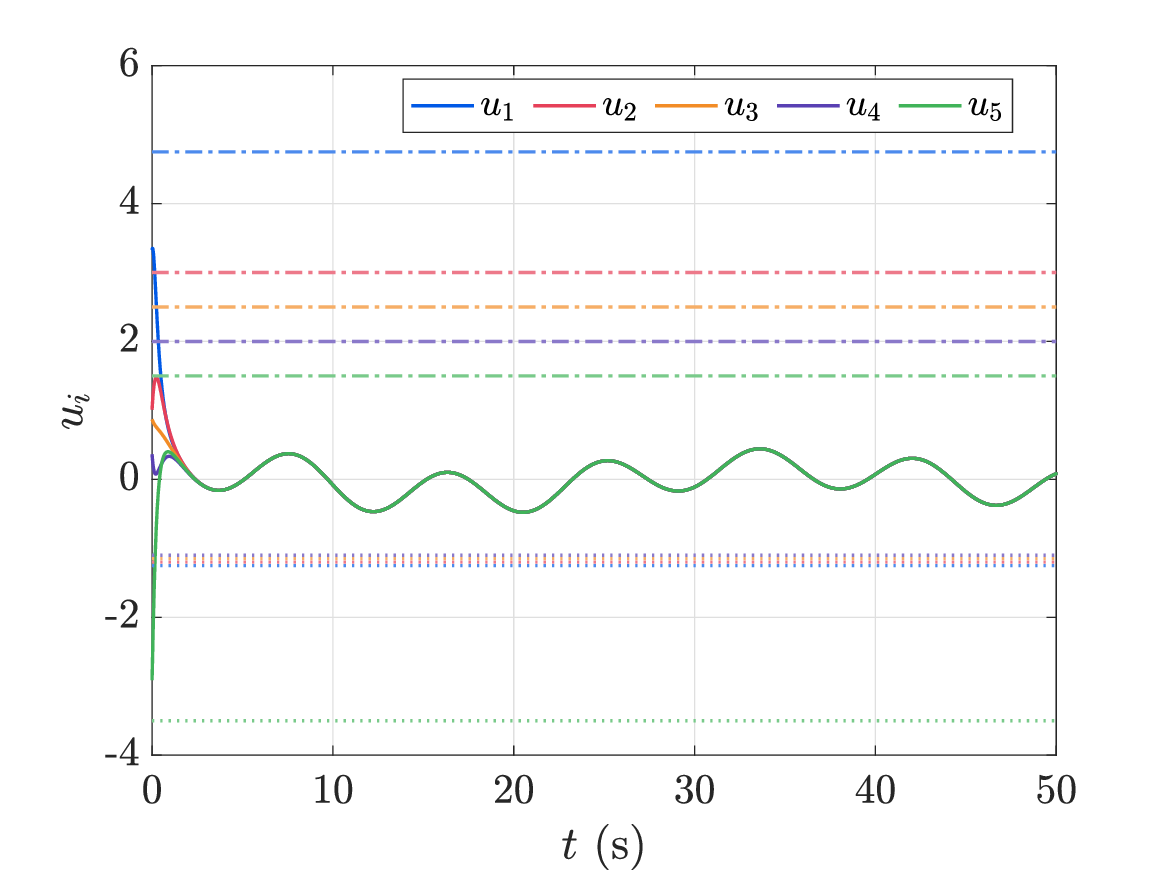}
    \caption{Realized actuator inputs $u_{i}(t)$.}
    \label{fig:ui_mf}
    \end{subfigure}
    \begin{subfigure}[t]{.325\linewidth}
        \centering
        \includegraphics[width=\linewidth]{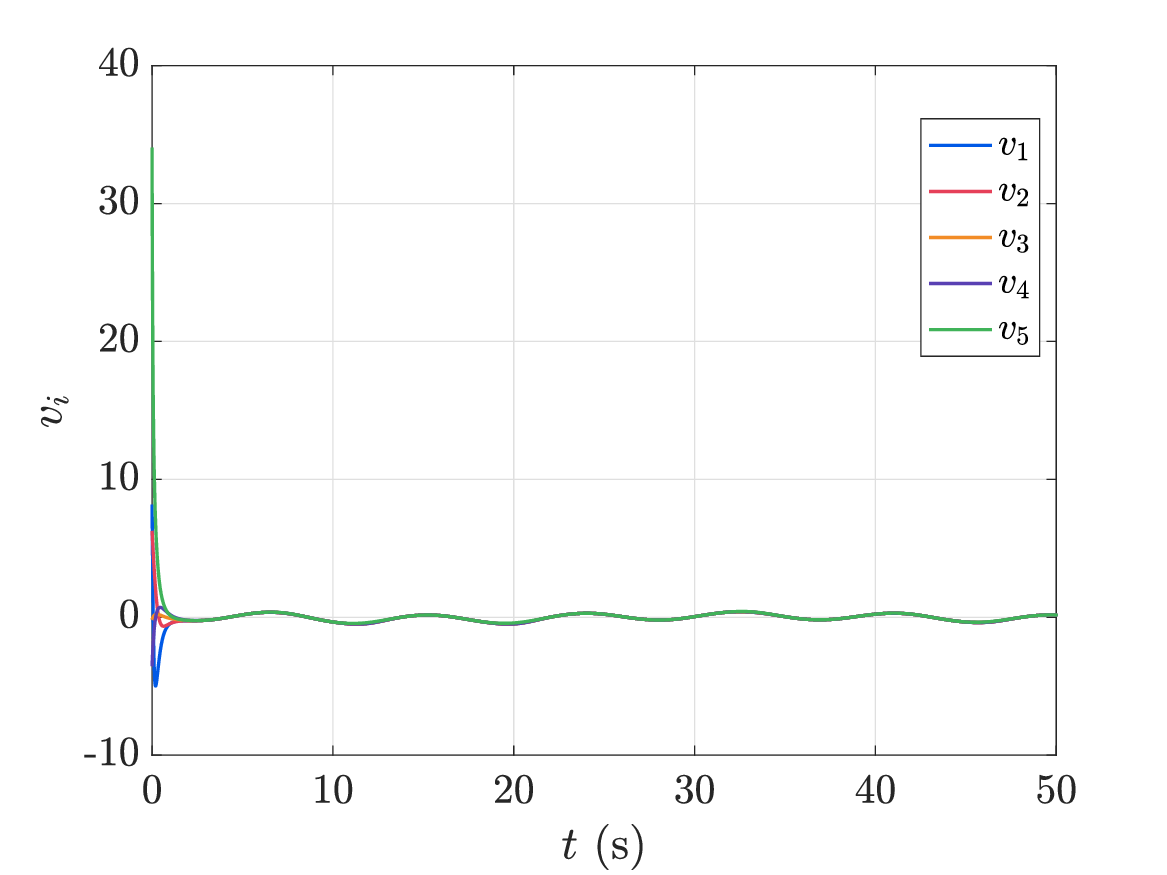}
    \caption{Commanded actuator signals $v_{i}(t)$.}
    \label{fig:vi_mf}
    \end{subfigure}
    \caption{Distributed safe-consensus for the multi-frequency prescribed reference case, generated via \eqref{eq:sim_rho_multifrequency}.}
\end{figure*}
\begin{figure*}[ht!]
    \centering
    \begin{subfigure}[t]{.325\linewidth}
        \centering
        \includegraphics[width=\linewidth]{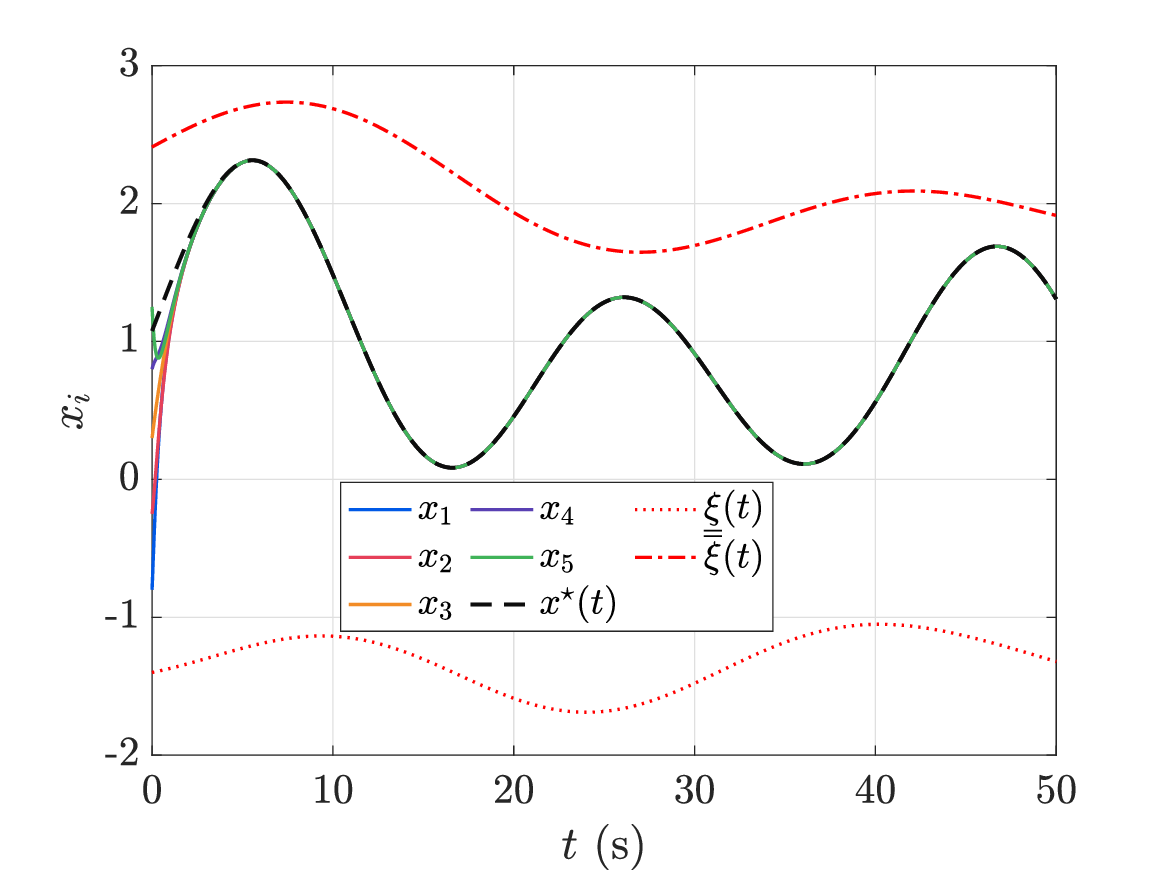}
    \caption{State trajectories $x_{i}(t)$.}
    \label{fig:xi_rs}
    \end{subfigure}
    \begin{subfigure}[t]{.325\linewidth}
        \centering
        \includegraphics[width=\linewidth]{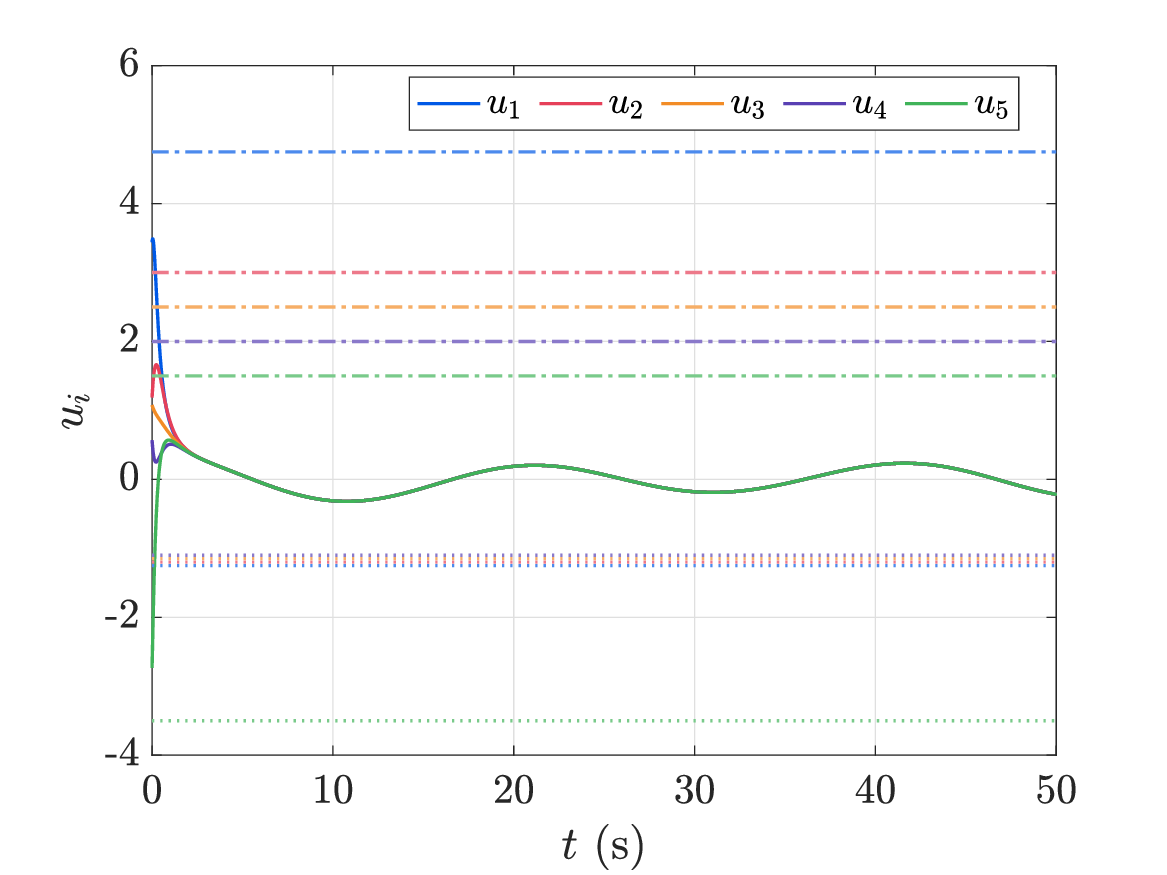}
    \caption{Realized actuator inputs $u_{i}(t)$.}
    \label{fig:ui_rs}
    \end{subfigure}
    \begin{subfigure}[t]{.325\linewidth}
        \centering
        \includegraphics[width=\linewidth]{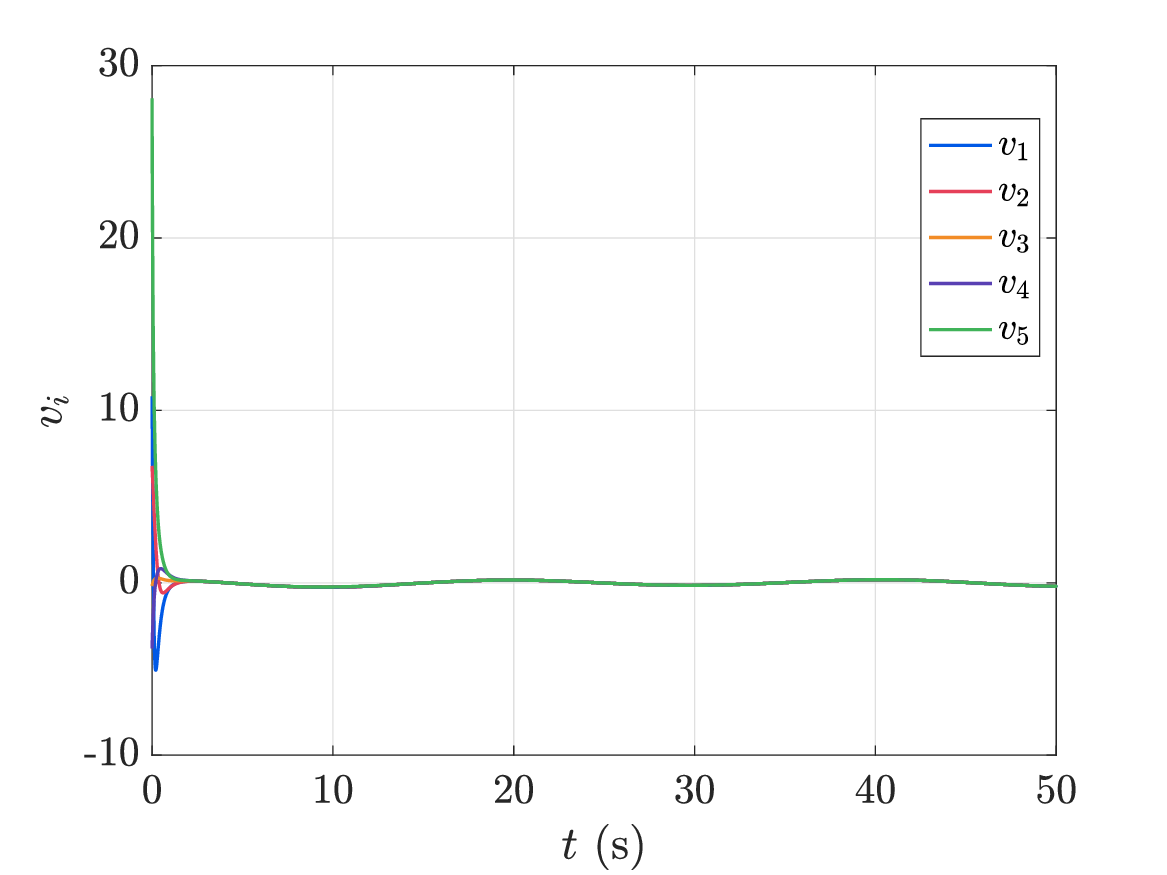}
    \caption{Commanded actuator signals $v_{i}(t)$.}
    \label{fig:vi_rs}
    \end{subfigure}
    \caption{Distributed safe-consensus for the biased sinusoidal prescribed reference case, generated via \eqref{eq:sim_rho_sine}.}
\end{figure*}
\begin{figure*}[ht!]
    \centering
    \begin{subfigure}[t]{.325\linewidth}
        \centering
        \includegraphics[width=\linewidth]{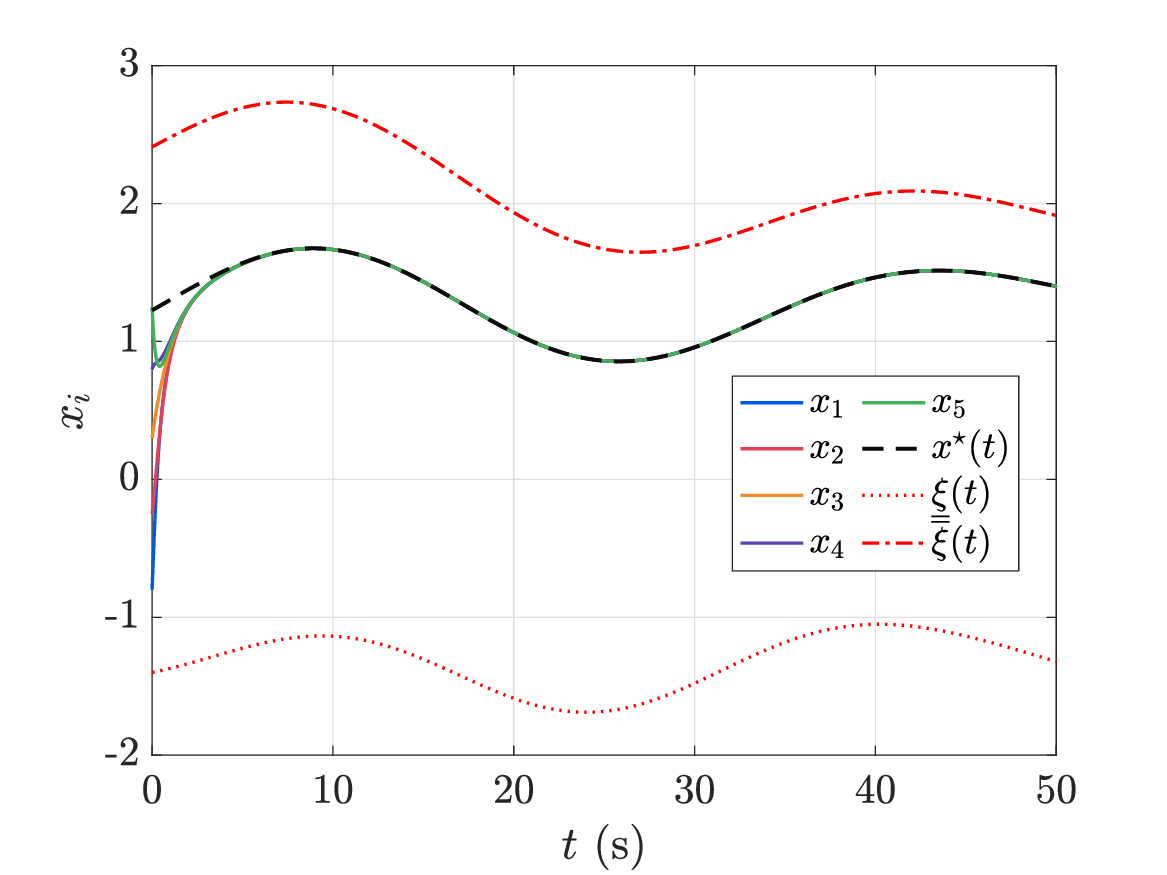}
    \caption{State trajectories $x_{i}(t)$.}
    \label{fig:xi_pr}
    \end{subfigure}
    \begin{subfigure}[t]{.325\linewidth}
        \centering
        \includegraphics[width=\linewidth]{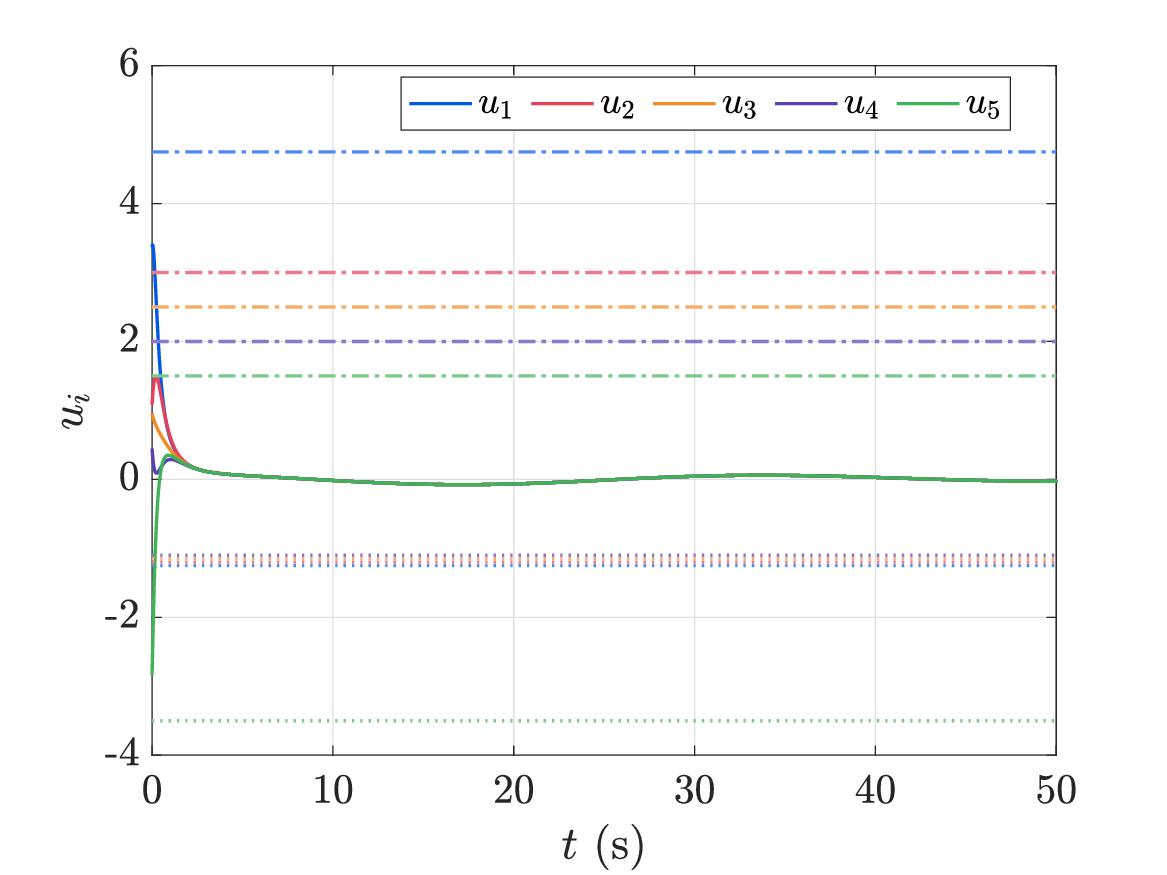}
    \caption{Realized actuator inputs $u_{i}(t)$.}
    \label{fig:ui_pr}
    \end{subfigure}
    \begin{subfigure}[t]{.325\linewidth}
        \centering
        \includegraphics[width=\linewidth]{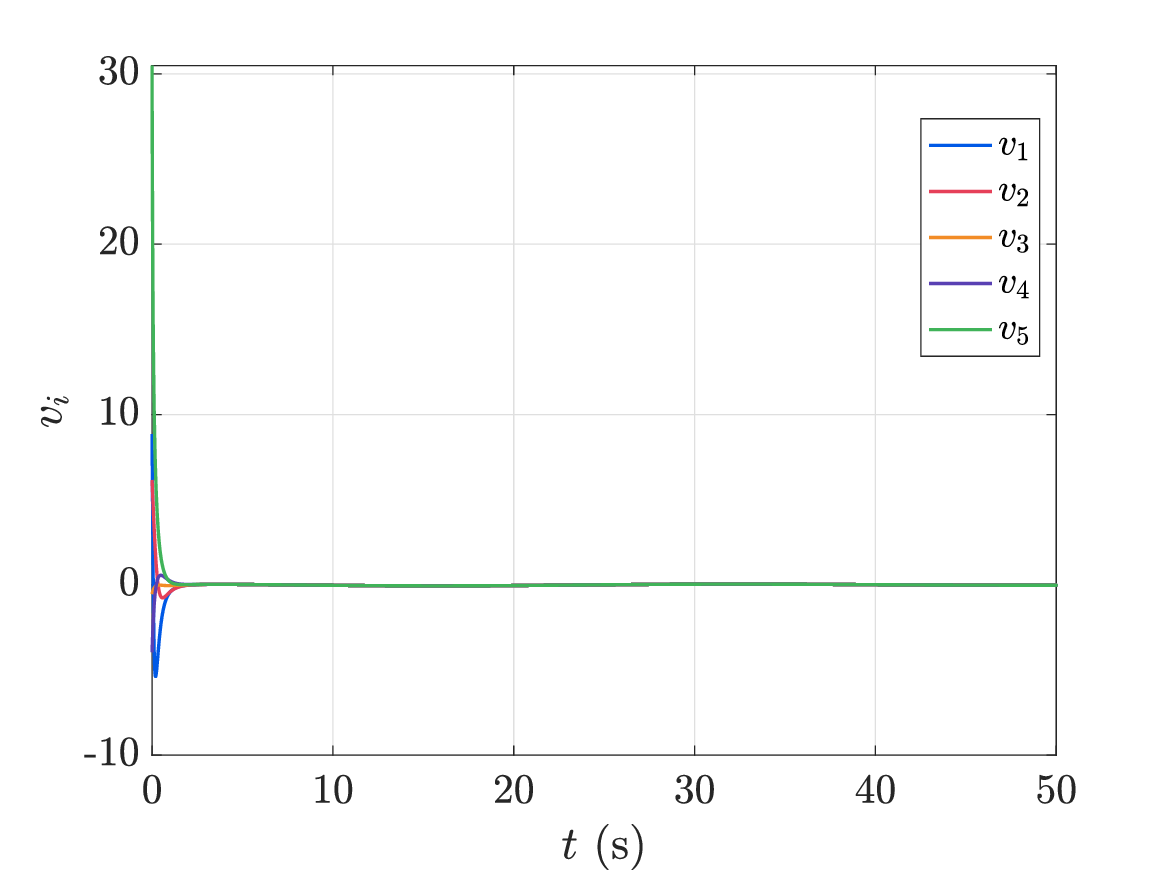}
    \caption{Commanded actuator signals $v_{i}(t)$.}
    \label{fig:vi_pr}
    \end{subfigure}
    \caption{Distributed safe-consensus for the smooth transition prescribed reference case, generated via \eqref{eq:sim_rho_piecewise}.}
\end{figure*}
\begin{figure*}[ht!]
    \centering
    \begin{subfigure}[t]{.325\linewidth}
        \centering
        \includegraphics[width=\linewidth]{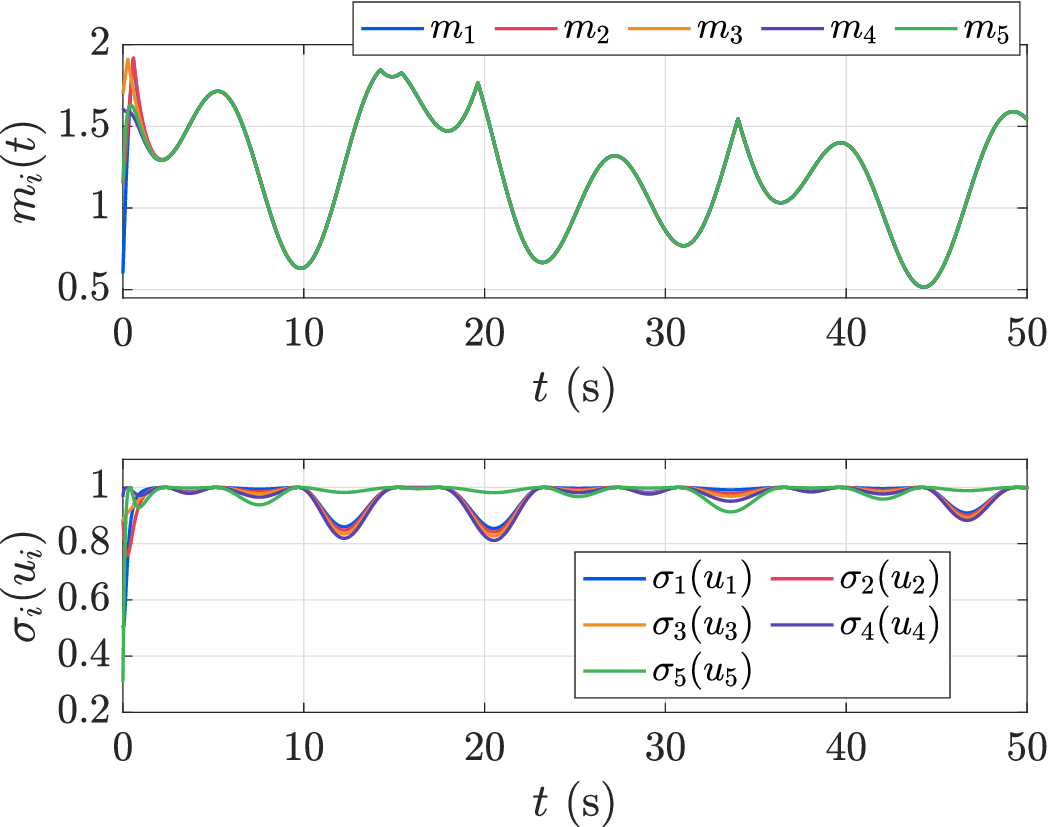}
    \caption{Multi-frequency prescribed reference.}
    \label{fig:sm_mf}
    \end{subfigure}
    \begin{subfigure}[t]{.325\linewidth}
        \centering
        \includegraphics[width=\linewidth]{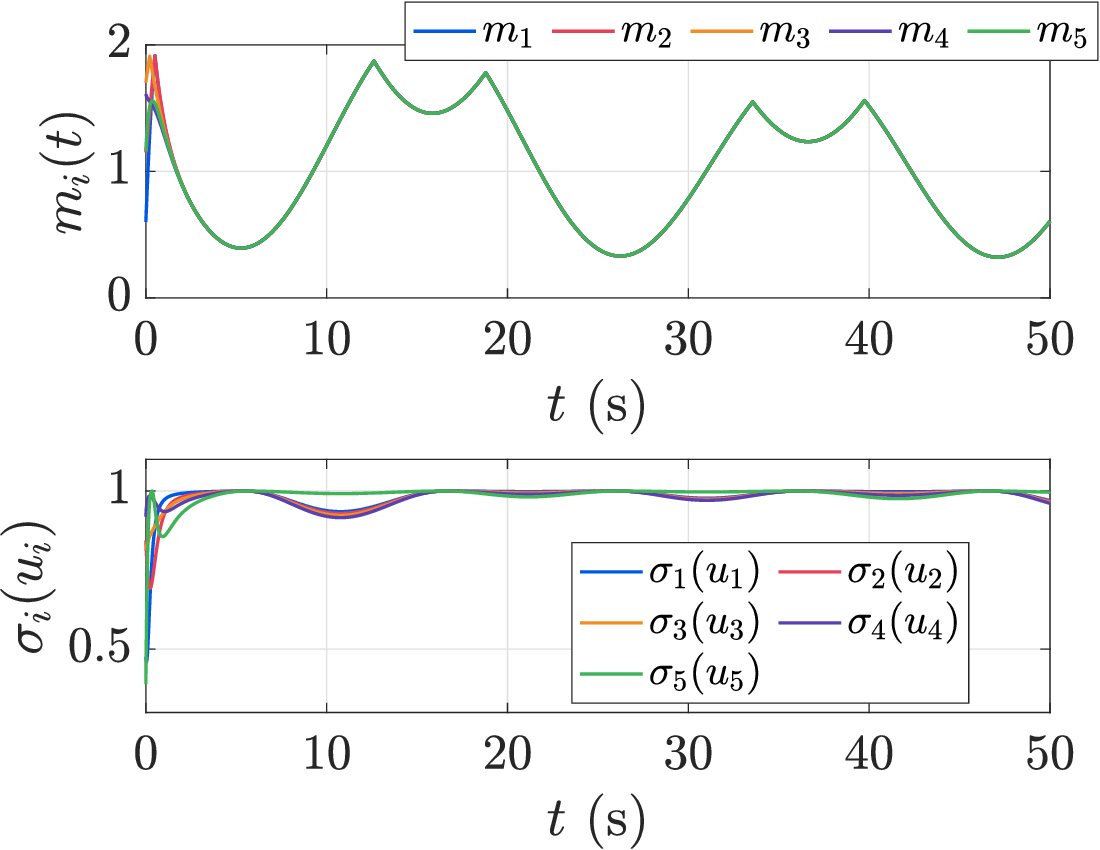}
    \caption{Biased sinusoidal prescribed reference.}
    \label{fig:sm_rs}
    \end{subfigure}
    \begin{subfigure}[t]{.325\linewidth}
        \centering
        \includegraphics[width=\linewidth]{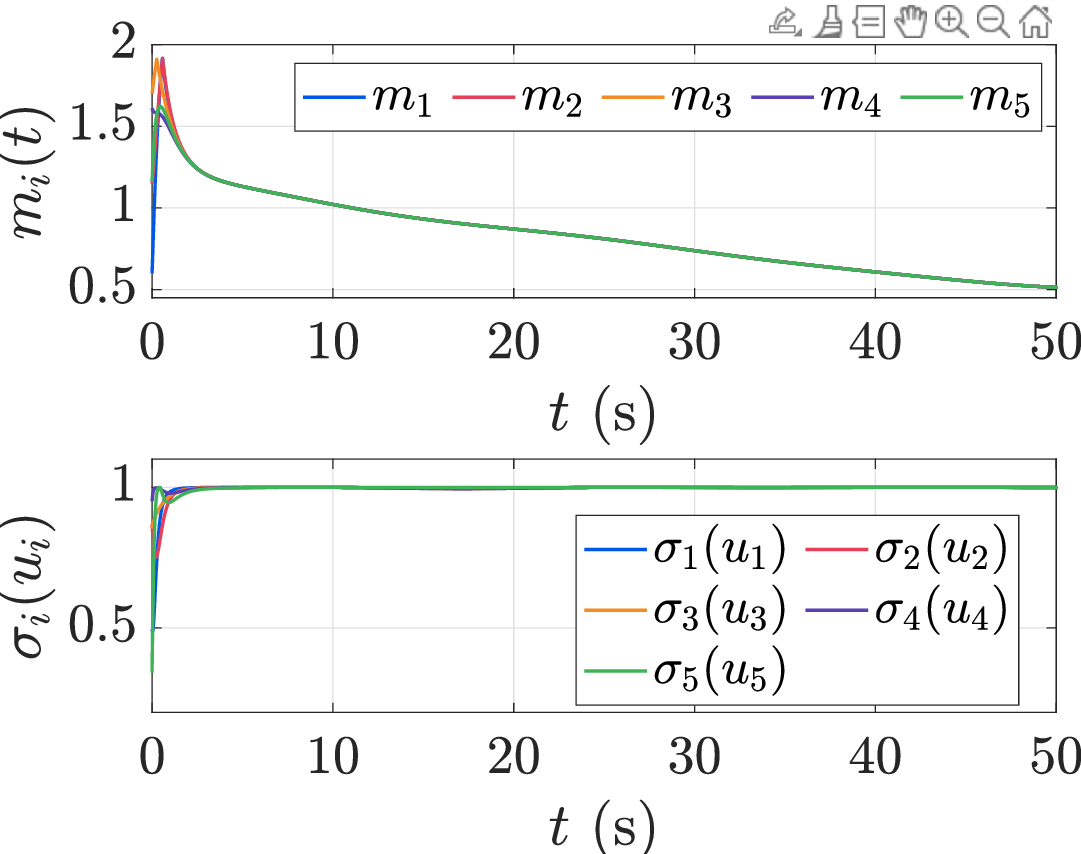}
    \caption{Smooth transition prescribed reference.}
    \label{fig:sm_pr}
    \end{subfigure}
    \caption{Safety margins and actuator regularity factors for the prescribed-trajectory synchronization problem.}
    \label{fig:sm}
\end{figure*}
\begin{figure*}[ht!]
    \centering
    \begin{subfigure}[t]{.325\linewidth}
        \centering
        \includegraphics[width=\linewidth]{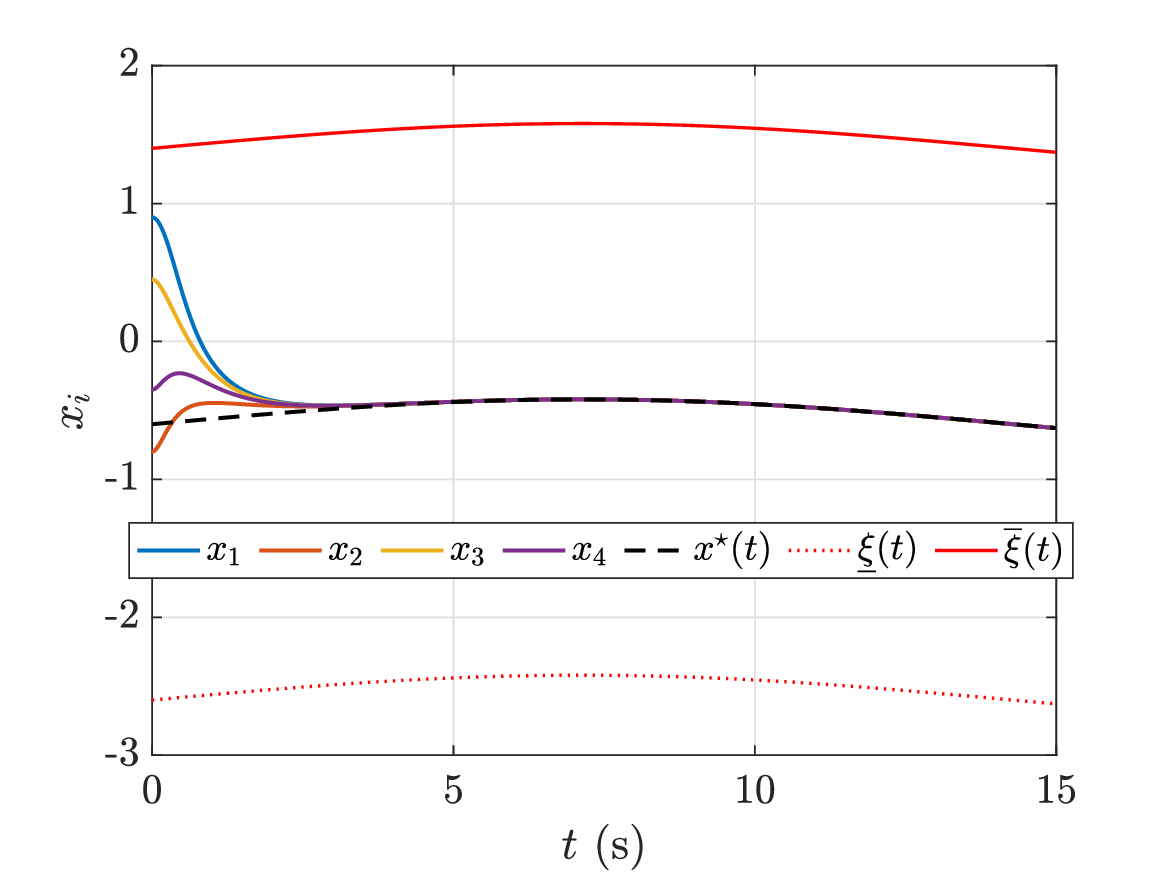}
    \caption{State trajectories $x_{i}(t)$.}
    \label{fig:tvsafe_states_short}
    \end{subfigure}
    \begin{subfigure}[t]{.325\linewidth}
        \centering
        \includegraphics[width=\linewidth]{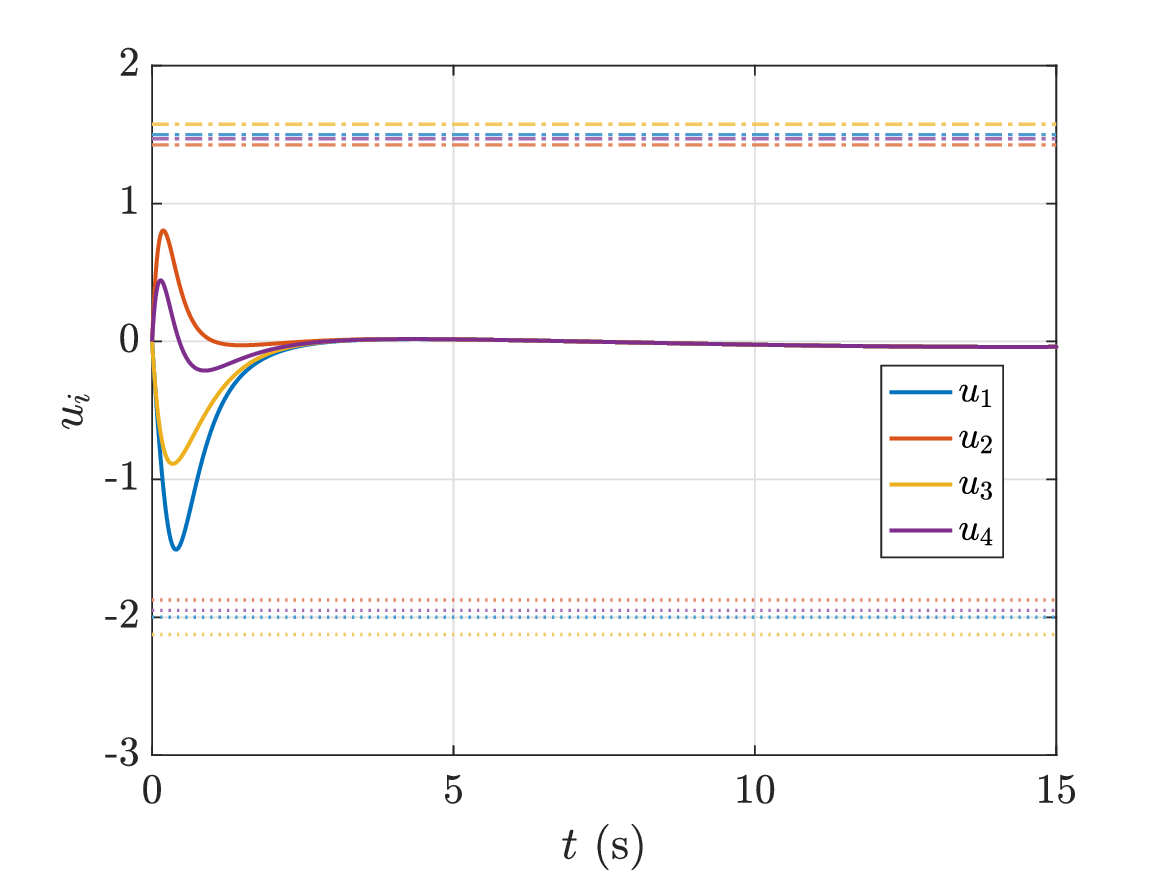}
    \caption{Realized actuator inputs $u_{i}(t)$.}
    \label{fig:tvsafe_inputs_short}
    \end{subfigure}
    \begin{subfigure}[t]{.325\linewidth}
        \centering
        \includegraphics[width=\linewidth]{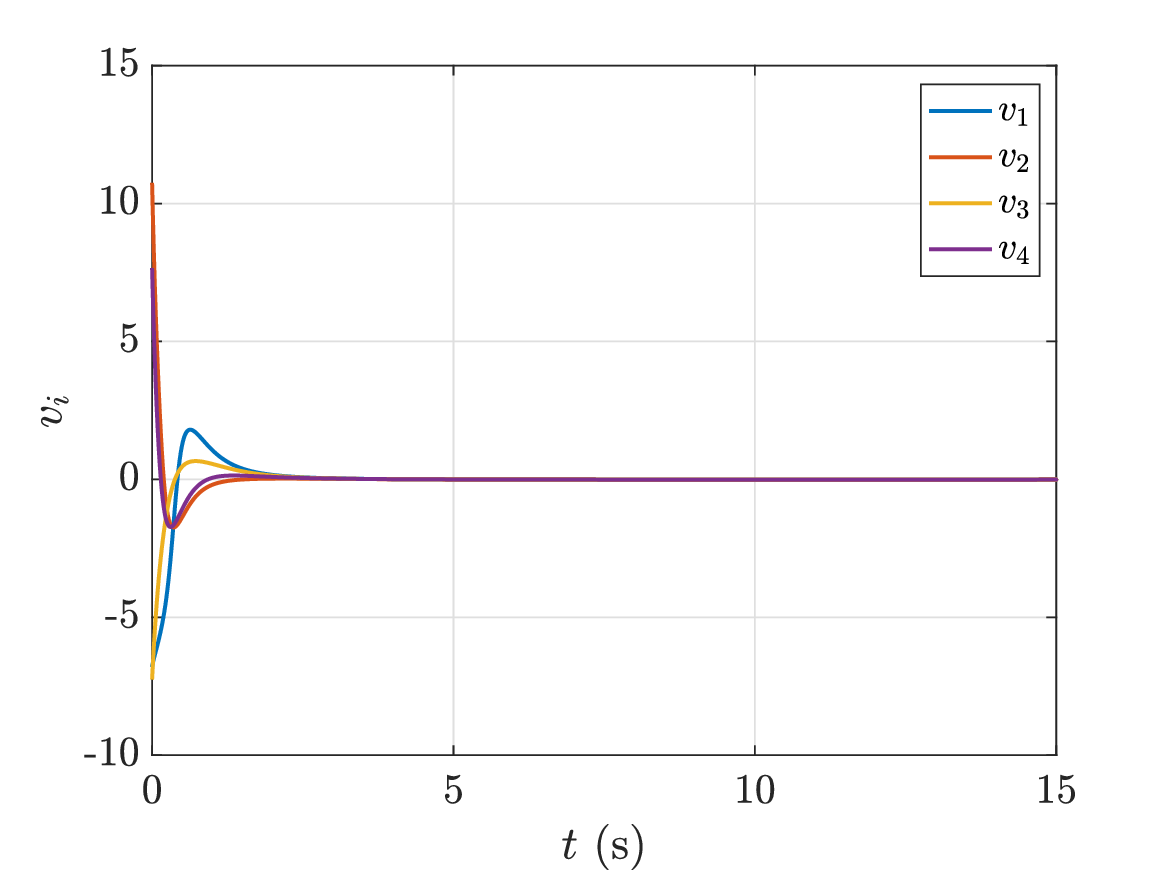}
    \caption{Commanded actuator signals $v_{i}(t)$.}
    \label{fig:tvsafe_command_short}
    \end{subfigure}
    \caption{Distributed safe-consensus for the constant barrier coordinate reference.}
\end{figure*}

We now illustrate the proposed prescribed-trajectory safe-consensus strategy for a network of $N=5$ single-integrator agents communicating over a connected undirected cycle graph. The controller gains are selected as $k_{z}=1.20$, $\kappa=1.00$, and $c_{i}=3.00$. The actuator parameters are chosen as $p_{1,i}=1$, $p_{2,i}=0.70$, and $\gamma_{i}=2$ for all $i\in \mathcal{V}$. To demonstrate heterogeneous and asymmetric actuator authority, the admissible actuator intervals are selected as $\underline{\mathbf{u}}
    =
    \begin{bmatrix}
        -1.25 & -1.20 & -1.15 & -1.10 & -3.50
    \end{bmatrix}^{\top}$ and $\overline{\mathbf{u}}
    =
    \begin{bmatrix}
        4.75 & 3.00 & 2.50 & 2.00 & 1.50
    \end{bmatrix}^{\top}$. Thus, each agent has a different admissible input interval, and the positive and negative actuation limits are not symmetric. In the plots, upper bounds are in dash-dot, whereas lower bounds are in dots. The common admissible core \eqref{eq:Omega} is selected as a time-varying asymmetric interval, where
\begin{align}
    \underline{\xi}(t)
    =
    c_{\xi}(t)-\Delta_{\ell}(t),
    ~~
    \overline{\xi}(t)
    =
    c_{\xi}(t)+\Delta_{u}(t),
    \label{eq:sim_core_bounds}
\end{align}
with
\begin{align}
    c_{\xi}(t)
    =&
    0.25\sin{\left(0.18t\right)}
    +
    0.10\cos{\left(0.07t\right)},\\
    \Delta_{\ell}(t)
    =&
    1.35
    +
    0.18\sin{\left(0.11t+0.40\right)}
    +
    0.08\cos{\left(0.27t\right)},\\
    \Delta_{u}(t)
    =&
    2.10
    +
    0.22\cos{\left(0.09t-0.30\right)}
    +
    0.10\sin{\left(0.21t\right)}.
\end{align}
This choice gives different lower and upper clearances from the nominal signal $c_{\xi}(t)$. In particular, $\Delta_{\ell}(t)\ge 1.09$ and $\Delta_{u}(t)\ge 1.78$, so the common core has a uniformly positive width. The prescribed admissible trajectory is generated through a normalized location variable $\rho(t)$ inside the moving interval. Specifically, we set
\begin{align}
    x^{\star}(t)
    =
    \underline{\xi}(t)
    +
    \rho(t)
    \left(
        \overline{\xi}(t)-\underline{\xi}(t)
    \right),
    \label{eq:sim_xstar}
\end{align}
where $\rho(t)\in(0,1)$ specifies the relative location of $x^{\star}(t)$ within the admissible core \eqref{eq:Omega}. This construction is useful because admissibility is enforced directly through the scalar signal $\rho(t)$. If there exist constants $\rho_{\min}$ and $\rho_{\max}$ such that $0
    <
    \rho_{\min}
    \le
    \rho(t)
    \le
    \rho_{\max}
    <
    1$ for all $t\ge 0$, then $x^{\star}(t)\in\Omega(t)$ for all $t\ge 0$. Moreover,
\begin{align}
    x^{\star}(t)-\underline{\xi}(t)
    =&
    \rho(t)
    \left(
        \overline{\xi}(t)-\underline{\xi}(t)
    \right),\\
    \overline{\xi}(t)-x^{\star}(t)
    =&
    \left(
        1-\rho(t)
    \right)
    \left(
        \overline{\xi}(t)-\underline{\xi}(t)
    \right).
\end{align}
Thus, using \eqref{eq:common_safe_core_width_main}, the prescribed trajectory remains uniformly separated from the two moving boundaries according to
\begin{align}
    \min
    \left\{
        x^{\star}(t)-\underline{\xi}(t),
        \overline{\xi}(t)-x^{\star}(t)
    \right\}
    \ge
    \min
    \left\{
        \rho_{\min},
        1-\rho_{\max}
    \right\}
    \delta_{\xi}.
    \label{eq:sim_xstar_margin}
\end{align}
This allows one to test different reference patterns while preserving the structural admissibility required by the theory.

Three representative choices of $\rho(t)$ are considered. The first one is a multi-frequency reference,
\begin{align}
    \rho_{\mathrm{mf}}(t)
    =
    0.50
    +
    0.22\sin{\left(0.18t\right)}
    +
    0.12\sin{\left(0.73t+0.60\right)}.
    \label{eq:sim_rho_multifrequency}
\end{align}
The signal \eqref{eq:sim_rho_multifrequency} combines a slow component and a faster oscillatory component. It is chosen to test whether the controller can track a nontrivial prescribed trajectory that does not move monotonically and contains multiple time scales. Since $0.50-0.22-0.12
    \le
    \rho_{\mathrm{mf}}(t)
    \le
    0.50+0.22+0.12$, one has $0.16
    \le
    \rho_{\mathrm{mf}}(t)
    \le
    0.84,~\forall t\ge 0$. Therefore, the prescribed trajectory generated by \eqref{eq:sim_rho_multifrequency} remains strictly inside $\Omega(t)$ with a normalized boundary clearance of at least $0.16$.

The second choice is a single-frequency biased sinusoidal reference,
\begin{align}
    \rho_{\mathrm{s}}(t)
    =
    0.65
    +
    0.25\sin{\left(0.30t\right)}.
    \label{eq:sim_rho_sine}
\end{align}
This case provides a simpler periodic trajectory and is useful for isolating the effect of smooth oscillatory motion without additional high-frequency content. The offset $0.65$ biases the trajectory toward the upper portion of the admissible interval. This is intentional in the present simulation because the actuator limits are heterogeneous and asymmetric, with more restrictive negative authority for some agents. Since $0.40
    \le
    \rho_{\mathrm{s}}(t)
    \le
    0.90,
    ~
    \forall t\ge 0$, the corresponding trajectory remains admissible, with a normalized boundary clearance of at least $0.10$.

The third choice is a smooth transition-type reference,
\begin{align}
    \rho_{\mathrm{p}}(t)
    =
    0.75
    +
    0.08\tanh{\left(0.05\left(t-20\right)\right)}
    +
    0.02\sin{\left(0.15t\right)}.
    \label{eq:sim_rho_piecewise}
\end{align}
This case is included to emulate a commanded relocation of the desired operating point inside the safe corridor. The hyperbolic tangent term produces a smooth transition centered around $t=20$, while the sinusoidal term adds a small persistent variation after and during the transition. Since $-1
    <
    \tanh{\left(0.05\left(t-20\right)\right)}
    <
    1$, it follows that $0.75-0.08-0.02
    <
    \rho_{\mathrm{p}}(t)
    <
    0.75+0.08+0.02$, and hence $0.65
    <
    \rho_{\mathrm{p}}(t)
    <
    0.85,
    ~
    \forall t\ge 0$. Thus, this reference remains well inside the safe interval while testing the transient response to a deliberate shift in the prescribed admissible trajectory.

For each of the above cases, the physical prescribed trajectory $x^{\star}(t)$ is obtained from \eqref{eq:sim_xstar}, and the corresponding transformed reference is computed using \eqref{eq:z_star_main}. Equivalently, since $x^{\star}(t)$ is parameterized by $\rho(t)$, one may write
\begin{align}
    z^{\star}(t)
    =
    \ln{
    \left(
        \frac{\rho(t)}
        {1-\rho(t)}
    \right)}.
    \label{eq:sim_zstar_rho}
\end{align}
The expression in \eqref{eq:sim_zstar_rho} shows that the normalized reference location $\rho(t)$ determines the transformed reference used by the controller. The three cases above therefore test the same safe-synchronization mechanism under progressively different reference profiles (multi-frequency tracking, biased sinusoidal tracking, and smooth transition tracking).

The initial states and initial actuator-tracking errors are selected as $\mathbf{x}(0)
    =
    \begin{bmatrix}
        -0.80 & -0.25 & 0.30 & 0.80 & 1.25
    \end{bmatrix}^{\top}$ and 
    $\boldsymbol{\varepsilon}(0)
    =
    \begin{bmatrix}
        0.05 & -0.08 & 0.04 & -0.06 & 0.03
    \end{bmatrix}^{\top}$.

\Cref{fig:xi_mf,fig:xi_rs,fig:xi_pr} show the evolution of the agent states together with the corresponding user-selected prescribed trajectory $x^{\star}(t)$ and the asymmetric admissible bounds $\underline{\xi}(t)$ and $\overline{\xi}(t)$ for three different choices of $\rho(t)$. One may observe that all agent states remain strictly inside the time-varying safe core throughout the simulation and converge to the prescribed admissible trajectory. This also confirms that the agents do reach an agreement with each other and synchronize with the user-defined safe trajectory.

The realized actuator inputs are shown in \Cref{fig:ui_mf,fig:ui_rs,fig:ui_pr}. Despite the heterogeneous and asymmetric input intervals, the realized inputs remain strictly inside their corresponding admissible sets. The safety and actuator-regularity plots further confirm the forward-invariance and nonsingularity properties of the proposed closed-loop system. The state-safety margin $m_i(t)
    :=
    \min
    \left\{
        x_i(t)-\underline{\xi}(t),
        \overline{\xi}(t)-x_i(t)
    \right\}$ remains strictly positive for every agent (see \Cref{fig:sm}), which shows that no trajectory approaches the boundary of the moving admissible core $\Omega(t)$. Hence, the logarithmic barrier coordinates remain well defined over the entire simulation horizon. Similarly, the actuator regularization factor $\sigma_i(u_i(t))$ stays positive for all agents, which implies that the realized inputs remain strictly inside their heterogeneous asymmetric admissible intervals. Since the commanded input $v_i$ contains $\sigma_i(u_i)$ in the denominator, the positivity of $\sigma_i(u_i(t))$ verifies that the actuator-compensating control law remains nonsingular. Consequently, the plots of $m_i(t)$ and $\sigma_i(u_i(t))$ provide a numerical confirmation that both state safety and actuator admissibility are preserved while the agents synchronize to the prescribed admissible trajectory.

These results confirm the three main conclusions of the proposed theory: forward invariance of the prescribed safe set, strict admissibility of the actuator inputs, and asymptotic synchronization to $x^{\star}(t)$. Moreover, the simulation highlights two important aspects of the proposed design. First, the prescribed trajectory is not required to be constant or centered inside the safe interval. It may vary with multiple frequencies as long as it remains strictly admissible. Second, the admissible actuator intervals may be both asymmetric and agent-dependent. The barrier-coordinate transformation enforces state safety, while the transformed pinning term drives all agents to the prescribed safe trajectory in the transformed coordinate system.

As a special case, we now consider a network of $N=4$ agents communicating over a connected undirected cycle graph, and elucidate the results in \Cref{cor:fixed_transformed_reference_main}. The barrier-consensus gains are chosen as $k_{z} = 1.4, \kappa = 1.2$, and the actuator-tracking gains are selected as $c =
    \begin{bmatrix}
        2 & 2.1 & 2.2 & 2.3
    \end{bmatrix}^{\top}$. For the continuously differentiable asymmetric actuator dynamics, we choose $p_{1,i} = 1, p_{2,i} = 0.25,  \gamma_{i} = 4$ (for all $i$) with asymmetric actuator limits $\underline{\mathbf{u}}
    =
    \begin{bmatrix}
        -2.0 & -1.875 & -2.125 & -1.95
    \end{bmatrix}^{\top}$, and $\overline{\mathbf{u}}
    =
    \begin{bmatrix}
        1.5 & 1.425 & 1.575 & 1.47
    \end{bmatrix}^{\top}$. The common time-varying safe interval is chosen as $\underline{\xi}(t) = -2.6 + 0.18\sin(0.22t)$, $\overline{\xi}(t) = \phantom{-}1.4 + 0.18\sin(0.22t)$, so that $\overline{\xi}(t)-\underline{\xi}(t)=4.0~\forall t \ge 0$. The desired barrier-coordinate equilibrium is fixed at $z^{\star}=0$, which yields the admissible synchronized reference $x^{\star}(t)
    =
    \frac{
        \underline{\xi}(t)+\overline{\xi}(t)
    }{2}$. The initial conditions are taken as $x(0)
    =
    \begin{bmatrix}
        0.9 & -0.8 & 0.45 & -0.35
    \end{bmatrix}^{\top}$, $u(0)
    =
    \begin{bmatrix}
        0 & 0 & 0 & 0
    \end{bmatrix}^{\top}$, which lie strictly inside the admissible state and actuator sets.  \Cref{fig:tvsafe_states_short} shows the state trajectories together with the time-varying safety bounds and the reference $x^{\star}(t)$. All states remain strictly inside the admissible interval and converge to the moving safe trajectory. \Cref{fig:tvsafe_inputs_short} displays the actuator outputs and confirms strict satisfaction of the asymmetric bounds. \Cref{fig:tvsafe_command_short} depicts the actual commanded actuator signals. The simulation confirms that the proposed controller simultaneously enforces forward invariance of the time-varying safe set, strict satisfaction of asymmetric actuator constraints, and asymptotic synchronization to the admissible reference trajectory.

\section{Conclusions}\label{sec:conclusion}
This paper investigated safe distributed consensus for single-integrator multi-agent systems over connected undirected graphs under simultaneous asymmetric actuator constraints and time-varying state safety constraints. The proposed framework yields a closed-loop architecture that simultaneously accounts for both actuation limits and safe-state requirements. For compact admissible sets of initial conditions, we established completeness of solutions, boundedness of all closed-loop signals, strict satisfaction of asymmetric input bounds, forward invariance of the prescribed safe interval, exponential convergence of the actuator-tracking errors, and asymptotic synchronization of the agent states to a designer-selected admissible trajectory embedded in the common safe set. The present results provide initial steps to advance safety-critical networked autonomy via concurrent network-level coordination and actuator-level realization.

\bibliographystyle{ieeetr}
\bibliography{references}

\begin{thebibliography}{10}

\bibitem{OlfatiSaber2007}
R.~Olfati-Saber, J.~A. Fax, and R.~M. Murray, ``Consensus and cooperation in
  networked multi-agent systems,'' {\em Proceedings of the IEEE}, vol.~95,
  no.~1, pp.~215--233, 2007.

\bibitem{OlfatiSaberMurray2004}
R.~Olfati-Saber and R.~M. Murray, ``Consensus problems in networks of agents
  with switching topology and time-delays,'' {\em IEEE Transactions on
  Automatic Control}, vol.~49, no.~9, pp.~1520--1533, 2004.

\bibitem{RenBeard2005}
W.~Ren and R.~W. Beard, ``Consensus seeking in multi-agent systems under
  dynamically changing interaction topologies,'' {\em IEEE Transactions on
  Automatic Control}, vol.~50, no.~5, pp.~655--661, 2005.

\bibitem{mesbahi2010graph}
M.~Mesbahi and M.~Egerstedt, {\em Graph Theoretic Methods in Multiagent
  Networks}.
\newblock Princeton University Press, 2010.

\bibitem{9072294}
S.-H. Kwon, Y.-B. Bae, J.~Liu, and H.-S. Ahn, ``From matrix-weighted consensus
  to multipartite average consensus,'' {\em IEEE Transactions on Control of
  Network Systems}, vol.~7, no.~4, pp.~1609--1620, 2020.

\bibitem{Ames2019}
A.~D. Ames, S.~Coogan, M.~Egerstedt, G.~Notomista, K.~Sreenath, and P.~Tabuada,
  ``Control barrier functions: Theory and applications,'' {\em 2019 18th
  European Control Conference (ECC)}, pp.~3420--3431, 2019.

\bibitem{Wang2017}
L.~Wang, A.~D. Ames, and M.~Egerstedt, ``Safety barrier certificates for
  collisions-free multirobot systems,'' {\em IEEE Transactions on Robotics},
  vol.~33, no.~3, pp.~661--674, 2017.

\bibitem{actuator_constraints_1}
R.~Rajamani, {\em Vehicle Dynamics and Control}.
\newblock Springer, 2~ed., 2012.

\bibitem{actuator_constraints_2}
T.~I. Fossen, {\em Handbook of Marine Craft Hydrodynamics and Motion Control}.
\newblock John Wiley \& Sons, 2~ed., 2021.

\bibitem{actuator_constraints_3}
B.~L. Stevens, F.~L. Lewis, and E.~N. Johnson, {\em Aircraft Control and
  Simulation: Dynamics, Controls Design, and Autonomous Systems}.
\newblock John Wiley \& Sons, 3~ed., 2016.

\bibitem{Tee2009}
K.~P. Tee, S.~S. Ge, and E.~H. Tay, ``Barrier lyapunov functions for the
  control of output-constrained nonlinear systems,'' {\em Automatica}, vol.~45,
  no.~4, pp.~918--927, 2009.

\bibitem{safe_coordination_2}
X.~Sun and C.~G. Cassandras, ``Optimal dynamic formation control of multi-agent
  systems in constrained environments,'' {\em Automatica}, vol.~73,
  pp.~169--179, 2016.

\bibitem{ranjan2026engagement}
P.~K. Ranjan, A.~Sinha, and Y.~Cao, ``Engagement-zone-aware input-constrained
  guidance for safe target interception in contested environments,'' {\em arXiv
  preprint arXiv:2603.23649}, 2026.

\bibitem{9149651}
Y.-H. Liu, C.-Y. Su, and H.~Li, ``Adaptive output feedback funnel control of
  uncertain nonlinear systems with arbitrary relative degree,'' {\em IEEE
  Transactions on Automatic Control}, vol.~66, no.~6, pp.~2854--2860, 2021.

\bibitem{zhang2025tcns_state_constraints}
Z.~Zhang, Z.~Chen, and W.~Fang, ``Cooperative fault tolerant control with state
  constraints for nonlinear multiple-agent systems,'' {\em IEEE Transactions on
  Control of Network Systems}, vol.~12, no.~2, pp.~1264--1276, 2025.

\bibitem{8950207}
Y.~Lv, J.~Fu, G.~Wen, T.~Huang, and X.~Yu, ``On consensus of multiagent systems
  with input saturation: Fully distributed adaptive antiwindup protocol design
  approach,'' {\em IEEE Transactions on Control of Network Systems}, vol.~7,
  no.~3, pp.~1127--1139, 2020.

\bibitem{10795196}
Y.~Xu and F.~Jabbari, ``Discrete-time leader-following multiagent systems:
  Saturation constraints and event-triggered control,'' {\em IEEE Transactions
  on Control of Network Systems}, vol.~12, no.~2, pp.~1354--1368, 2025.

\bibitem{Yang2014}
T.~Yang, Z.~Meng, D.~V. Dimarogonas, and K.~H. Johansson, ``Global consensus
  for discrete-time multi-agent systems with input saturation constraints,''
  {\em Automatica}, vol.~50, no.~2, pp.~499--506, 2014.

\bibitem{Yi2019}
X.~Yi, T.~Yang, J.~Wu, and K.~H. Johansson, ``Distributed event-triggered
  control for global consensus of multi-agent systems with input saturation,''
  {\em Automatica}, vol.~100, pp.~1--9, 2019.

\bibitem{zhang2023neural}
S.~Zhang, K.~Garg, and C.~Fan, ``Neural graph control barrier functions guided
  distributed collision-avoidance multi-agent control,'' in {\em 7th Annual
  Conference on Robot Learning}, 2023.

\bibitem{SAJJADIKIA2013762}
S.~Sajjadi-Kia and F.~Jabbari, ``Controllers for linear systems with bounded
  actuators: Slab scheduling and anti-windup,'' {\em Automatica}, vol.~49,
  no.~3, pp.~762--769, 2013.

\bibitem{10982103}
X.~Min, S.~Baldi, Y.~Shi, and W.~Yu, ``Safe control of multiagent systems via
  low-complexity control barrier functions,'' {\em IEEE Transactions on
  Automatic Control}, vol.~70, no.~10, pp.~6831--6845, 2025.

\bibitem{sharifi2023tcns_hocbf_leader_follower}
M.~Sharifi and D.~V. Dimarogonas, ``Higher order barrier certificates for
  leader-follower multi-agent systems,'' {\em IEEE Transactions on Control of
  Network Systems}, vol.~10, no.~2, pp.~900--911, 2023.

\bibitem{song2025tcns_safety_flocking}
Y.~Song, N.~P. Nguyen, H.-S. Park, Y.~You, M.~Lee, and H.~Oh, ``Distributed
  safety-critical optimal flocking control algorithm with feasibility
  enhancement of high-order control barrier function,'' {\em IEEE Transactions
  on Control of Network Systems}, vol.~12, no.~3, pp.~2052--2063, 2025.

\bibitem{10794772}
M.~Charitidou and D.~V. Dimarogonas, ``Virtual leader and distance-based
  formation control with funnel constraints,'' {\em IEEE Transactions on
  Control of Network Systems}, vol.~12, no.~2, pp.~1342--1353, 2025.

\bibitem{10556645}
P.~Mestres, C.~Nieto-Granda, and J.~Cortés, ``Distributed safe navigation of
  multi-agent systems using control barrier function-based controllers,'' {\em
  IEEE Robotics and Automation Letters}, vol.~9, no.~7, pp.~6760--6767, 2024.

\bibitem{11499442}
X.~Huang and L.~Long, ``Distributed asymptotic consensus safety-critical
  control of multi-agent systems with vector control barrier functions,'' {\em
  IEEE Transactions on Control of Network Systems}, pp.~1--10, 2026.

\bibitem{Kumar2025}
S.~Kumar, S.~R. Kumar, and A.~Sinha, ``Provably safe control for constrained
  nonlinear systems with bounded input,'' {\em arXiv preprint
  arXiv:2504.11592}, 2025.

\bibitem{chen2020tcns_prescribed_performance}
F.~Chen and D.~V. Dimarogonas, ``Leader--follower formation control with
  prescribed performance guarantees,'' {\em IEEE Transactions on Control of
  Network Systems}, vol.~8, no.~1, pp.~450--461, 2021.

\end{thebibliography}

\end{document}